\documentclass[11pt]{article}
\usepackage[utf8]{inputenc}

\usepackage[margin=1in]{geometry}
\usepackage[T1]{fontenc}
\usepackage{bbm, bm}
\usepackage{graphicx}
\usepackage{color, xcolor}
\usepackage{amsmath, amsfonts, amssymb, amsthm, thmtools, mathtools}
\usepackage{algorithm, algpseudocode}
\usepackage{pifont}

\PassOptionsToPackage{hyphens}{url}
\usepackage[colorlinks=true, allcolors=blue]{hyperref}
\usepackage[capitalise,nameinlink]{cleveref}

\usepackage[style=alphabetic, backend=biber, minalphanames=3, maxalphanames=4, maxbibnames=99, maxcitenames=99]{biblatex}

\crefformat{section}{#2\S#1#3}
\crefformat{subsection}{#2\S#1#3}
\crefformat{subsubsection}{#2\S#1#3}
\Crefformat{section}{#2\S#1#3}
\Crefformat{subsection}{#2\S#1#3}
\Crefformat{subsubsection}{#2\S#1#3}

\declaretheoremstyle[bodyfont=\it,qed=\qedsymbol]{noproofstyle}

\numberwithin{equation}{section}

\newcommand{\excl}{(\ding{72})}

\declaretheorem[name=Observation,numbered=no]{observation*}

\declaretheorem[numberlike=equation]{problem}

\declaretheorem[numberlike=equation]{theorem}

\declaretheorem[name=Theorem,numbered=no]{theorem*}

\declaretheorem[numberlike=equation]{lemma}
\declaretheorem[name=Lemma,numbered=no]{lemma*}

\declaretheorem[name=Corollary,numbered=no]{corollary*}

\declaretheorem[name=Proposition,numbered=no]{proposition*}

\declaretheorem[numberlike=equation]{claim}
\declaretheorem[name=Claim,numbered=no]{claim*}

\declaretheorem[name=Conjecture,numbered=no]{conjecture*}

\declaretheorem[name=Question,numbered=no]{question*}

\declaretheoremstyle[headfont=\normalfont\bfseries,bodyfont=\normalfont,postheadspace=1em,qed=$\lozenge$]{rmkstyle}

\declaretheorem[numberlike=equation,style=rmkstyle]{definition}
\declaretheorem[unnumbered,name=Definition,style=rmkstyle]{definition*}

\declaretheorem[unnumbered,name=Example,style=rmkstyle]{example*}

\declaretheorem[unnumbered,name=Notation=rmkstyle]{notation*}

\declaretheorem[unnumbered,name=Construction,style=rmkstyle]{construction*}
\usepackage{singer-macros}
\usepackage{mathtools}
\usepackage{tikz}
\usepackage{algorithm}
\usepackage{algpseudocode}
\title{Streaming Complexity Separations for Dense and Sparse Graphs}
\author{Anonymous Authors}
\author{Yang P. Liu\thanks{\texttt{yangl7@andrew.cmu.edu}. Carnegie Mellon University.} \and Hoai-An Nguyen\thanks{\texttt{hnnguyen@cs.cmu.edu}. Carnegie Mellon University.} \and Noah G. Singer\thanks{\texttt{ngsinger@cs.cmu.edu}. Carnegie Mellon University.} \and David P. Woodruff\thanks{\texttt{dwoodruf@cs.cmu.edu}. Carnegie Mellon University.}}
\date{}

\newcommand{\eps}{\varepsilon}
\renewcommand{\epsilon}{\varepsilon}

\newcommand{\expand}{\mathcal{E}}
\newcommand{\cut}{\mathcal{C}}

\newcommand{\Good}[2]{\mathsf{Good}^{#1}_{#2}}

\newcommand{\entropy}{\mathsf{H}}

\newcommand{\Unif}[1]{\mathrm{Unif}(#1)}
\newcommand{\MaxCut}{\textsc{Max-Cut}}
\newcommand{\MinCut}{\textsc{Min-Cut}}
\newcommand{\DenSub}{\textsc{Densest-Subgraph}}
\newcommand{\MaxCSP}{\textsc{Max-CSP}_{k,q}}

\newcommand{\Sim}{\textsc{Similarity}}
\newcommand{\Rare}[1]{\textsc{Rarity}_{#1}}

\newcommand{\val}[2]{\mathsf{val}_{#1} \parens*{ #2 }}
\newcommand{\opt}[1]{\mathsf{opt}_{#1}}
\newcommand{\den}[2]{\mathsf{den}_{#1} \parens*{ #2 }}
\newcommand{\maxden}[1]{\mathsf{maxden}_{#1}}
\newcommand{\maxval}[1]{\mathsf{opt}_{#1}}

\newcommand{\compS}{\overline{S}}
\newcommand{\compT}{\overline{T}}
\newcommand{\bip}[1]{\mathsf{Kb}(#1)}
\newcommand{\uprod}{\boxtimes}
\newcommand{\HBall}[2]{\mathsf{Ball}^{#1}(#2)}
\newcommand{\HPMBall}[2]{\mathsf{Ball}^{#1}_{\pm}(#2)}

\newcommand{\PermProb}{\textsc{Permutation-Index}}
\newcommand{\IdProb}{\textsc{Identity}}

\newcommand{\kom}[1]{\mathsf{K}(#1)}

\newcommand{\IndProb}{\textsc{Index}}
\newcommand{\vLeft}{\mathtt{L}}
\newcommand{\vRight}{\mathtt{R}}
\newcommand{\vTop}{\mathtt{T}}
\newcommand{\vBot}{\mathtt{B}}

\newcommand{\Disc}[3]{\delta_{#1 \mid #2}(#3)}
\newcommand{\Nbr}[2]{\mathrm{N}_{#1}(#2)}
\newcommand{\Cval}[3]{\mathrm{val}_{#1 \mid #2}(#3)}
\newcommand{\Copt}[2]{\mathrm{opt}_{#1 \mid #2}}
\newcommand{\Closs}[3]{\mathrm{loss}_{#1 \mid #2}(#3)}
\newcommand{\Cgood}[3]{\mathrm{Good}_{#1 \mid #2}^{#3}}

\newcommand{\Wss}{W_{\text{sink}\leftrightarrow\text{sink}}}
\newcommand{\Wsl}{W_{\text{sink}\leftrightarrow\text{left}}}
\newcommand{\Lsl}{\ell_{\text{sink}\leftrightarrow\text{left}}}
\newcommand{\Llr}{\ell_{\text{left}\leftrightarrow\text{right}}}

\newcommand{\GRR}[2]{\calG_{\text{right-regular}}(#1,#2)}
\newcommand{\slack}{\mathrm{slack}}

\newcommand{\Cden}{C_{\mathrm{den}}}
\newcommand{\Csep}{C_{\mathrm{sep}}}

\newcommand{\Tail}[2]{\mathrm{Tail}_{#1}(#2)}
\newcommand{\TailDeg}[2]{\mathrm{TailDeg}_{#1}(#2)}
\renewcommand{\deg}[2]{\mathrm{deg}_{#1}(#2)}

\newcommand{\adv}[3]{\mathrm{adv}_{#1\mid #2}(#3)}
\newcommand{\xA}[3]{\mu^1_{#1 \mid #2}(#3)}
\newcommand{\xB}[3]{\mu^2_{#1 \mid #2}(#3)}
\newcommand{\xCom}[3]{\mu^\cap_{#1 \mid #2}(#3)}

\newcommand{\Vtyp}{V_{\mathrm{typ}}}

\newcommand{\vecG}{\underline{G}}



\begin{document}

\maketitle

\newcommand{\hn}[1]{{\color{red} \textbf{Hoaian: #1}}}
\newcommand{\noah}[1]{{\color{blue} \textbf{Noah: #1}}}
\newcommand{\yang}[1]{{\color{purple} \textbf{Yang: #1}}}

\begin{abstract}
We identify a sharp separation in the streaming space complexity of \emph{Maximum Cut} when the algorithm must output an approximate cut (rather than only the approximate value). For dense graphs, we show that $O(n/\varepsilon^2)$ space is sufficient and that $\Omega(n)$ space is necessary. In contrast, for graphs with $\Theta(n/\varepsilon^2)$ edges, the situation is markedly different: we show that the problem requires $\Omega(n \log(\varepsilon^2 n)/\varepsilon^2)$ space for any $\varepsilon=\omega(1/\sqrt{n})$, which is tight for the full range of $\varepsilon$. We also give an $\Omega(n \log n/\varepsilon^2)$-space lower bound against deterministic algorithms for outputting a $(1-\varepsilon)$ approximation to the \emph{value} of the maximum cut.

Using similar techniques we prove an analogous sharp separation in the streaming space complexity of \emph{Densest Subgraph} and show that for every constant-arity \emph{CSP} over a constant-size alphabet and the \emph{Similarity} problem the space complexity in dense streams can be improved by shaving a logarithmic factor.
\end{abstract}

\newpage

\section{Introduction}
In this paper, we study the streaming complexity of graph optimization problems. 
Our focus is on understanding the space complexity required to output approximate solutions in a single pass over insertion-only streams.
In particular, we identify sharp separations for sparse and dense graphs for the streaming complexity of the fundamental $\MaxCut$ problem. In $\MaxCut$, the task is to output a cut, i.e., a partition of the vertices of the input graph, such that the value (the number of edges crossing the partition) is at least a $(1-\varepsilon)$ fraction of the maximum cut value. Our results also extend to other (hyper)graph optimization problems: constant-arity constraint satisfaction problems (CSPs) and $\DenSub$.

Prior works \cite{KK15,KKSV17,KK19,AKSY20,AN21,CKP+23,FMW25} have studied streaming algorithms for $(1/2+\eps)$-approximate $\MaxCut$ (albeit the value version) in both the single-pass and multi-pass settings. Motivated by this and connections to sparsification/space-compression of graph cuts (and more generally CSPs), we study the $\eps$-dependence of the space required for streaming $\MaxCut$ for $(1-\eps)$-approximations (values of $\eps$ bounded away from $1/2$). The simplest algorithm for streaming $\MaxCut$ simply samples a uniform subset of $O(n/\eps^2)$ edges in the stream. Standard concentration bounds prove that this provides an additive approximation to \emph{all cuts} in the graph, which implies a multiplicative approximation to $\MaxCut$ (which cuts at least half the edges). The space required by this algorithm is $O(\eps^{-2}n\log(n\eps^2))$, because there are $\binom{\binom{n}{2}}{n/\eps^2} \approx \exp(\Theta(\eps^{-2}n\log(n\eps^2)))$ graphs with $n/\eps^2$ edges.

Our first result (see \Cref{thm:main}) is that this bound is tight in general, and we construct a family of graphs with $n/\eps^2$ edges such that any streaming algorithm outputting a $(1-\eps)$-approximate $\MaxCut$ requires at least $\Theta(\eps^{-2}n\log(n\eps^2))$ bits of space. However, the case of dense graphs is different, and applying techniques from $F_0$ estimation gives an approximate $\MaxCut$ streaming algorithm in dense graphs that only uses $O(n/\eps^2)$ bits of space.
In other words, we exhibit a dichotomy for streaming $\MaxCut$ in sparse versus dense graphs.

\subsection{Our contributions}
In all of the following, the lower bounds hold against insertion-only one-pass streaming algorithms. 

We first give a lower bound against randomized algorithms for the problem of outputting an approximate maximum cut. In particular, the hard instance is a family of graphs that contain $n/\eps^2$ edges. Note that for constant $\eps$ this is a sparse graph. 
\begin{theorem} \label{thm:main}
Let $\epsilon_0 > 0$ be a sufficiently small constant.
Any randomized streaming algorithm that returns a $(1-\epsilon)$-approximate maximum cut in a graph with $n$ vertices with probability at least $0.51$, where $\epsilon(n) \in [\omega(1/\sqrt{n}), \epsilon_0]$, requires $\Omega(\epsilon^{-2} n \log(n \epsilon^2))$ bits of space.
\end{theorem}
This bound is tight for the range $\eps \in [\omega(1/\sqrt{n}), \epsilon_0]$ -- the matching upper bound follows from the fact that uniformly sampling about $n/\eps^2$ edges produces a graph whose \MaxCut~value is a good approximation, and that there are $\exp(\eps^{-2}n\log(n\eps^2))$ graphs with $n/\eps^2$ edges.

In the dense setting, we use the $F_0$ estimation algorithm of \cite{B2018optimal} to give an algorithm that returns a $(1-\eps)$-approximate maximum cut in $O(n/\eps^2)$ \emph{bits} of space in \emph{dense} graphs (\Cref{thm:maxcut:dense alg}). We also show that $\Omega(n)$ \emph{working space} (the space required by the algorithm not including the space for the output) is required for dense graphs. 

\begin{theorem}\label{thm:maxcut:dense hard}
For constant $\eps > 0$, any randomized streaming algorithm that returns a $(1-\eps)$-approximate maximum cut in a graph with $n$ vertices and $m =\Theta(n^2)$ edges with probability at least $2/3$ requires $\Omega(n)$ bits of working space. 
\end{theorem}

We give a lower bound against deterministic algorithms for the problem of outputting an approximation to the maximum cut \emph{value}. 
\begin{theorem} \label{thm:deterministic-max-cut}
For constant $\eps > 0$, any deterministic streaming algorithm which returns a $(1\pm \eps)$-approximation to the maximum cut value of a graph with $n$ vertices requires $\Omega(\eps^{-2}n \log n)$ bits of space. 
\end{theorem}

We also give separations for the problem of densest subgraph. We give an $\Omega(n \log n)$ lower bound for outputting both an approximation of the density of the densest subgraph as well as an approximate densest subgraph itself in sparse graphs (\Cref{thm:densub:sparse hard}). We then give an algorithm that returns an approximate densest subgraph in $O(n/\eps^2)$ bits of space in dense graphs (\Cref{thm:densub:dense alg}) and show that $\Omega(n)$ working space is necessary (\Cref{thm:densub:dense hard}). 

Using similar techniques, we show that for every constant-arity CSP over a constant size alphabet that space complexity $O(n/\eps^2)$ can be achieved in dense graphs (\Cref{thm:maxcsp:dense alg}). Finally, we give algorithms for Similarity (which interestingly doesn't require a dense stream) and Rarity in dense streams (\Cref{thm:sim-alg} and \Cref{thm:rare-alg}).

\subsection{Related Work}
\label{sec:related}

\paragraph{Cut sparsification.} We mention that there has been significant work on sparsifying graphs to preserve the value of every cut up to a \emph{multiplicative} $(1\pm\eps)$ (see e.g., \cite{AG09}); such algorithms are suited to the $\MinCut$ problem, rather than $\MaxCut$.
Indeed, \textcite{SW15} showed an $\Omega(n \log n)$-space lower bound for computing $\MinCut$ in sparse graphs based on a reduction from \textsc{Connectivity}; their lower bound can easily be extended to dense graphs by simply adding a highly dense subgraph connected to the existing graph by a single edge.
Thus, our $O(n)$-space algorithm for $\MaxCut$ (\Cref{thm:maxcut:dense alg}) proves that approximate $\MinCut$ is harder than approximate $\MaxCut$ in the dense streaming setting.

\paragraph{Densest subgraph.} In contrast to $\MaxCut$, polynomial-time algorithms are known for solving $\DenSub$ in the classical setting (see e.g., \cite{GGT89,Cha00}).
(For additional background on $\DenSub$, see the recent survey~\cite{LMFB24}.) Streaming algorithms for the $\DenSub$ problem have been studied in several previous works, including~\cite{BKV12,AKS+14,BHNT15,MTVV15,MP24,Md20}.
Most relevant to us is the work of \textcite{EHW16}, who gave an $O_\eps(n \log^2n)$-space algorithm for outputting a $(1-\eps)$-approximate densest subgraph in an arbitrary input graph based on a sparsification reduction.
Our \Cref{thm:densub:dense alg} improves over their algorithm by eliminating the $\log^2 n$ factor in the case of dense input graphs. \textcite{BKV2012densest} proved a similar bound to our \Cref{thm:densub:dense hard}, proving that any $p$-pass insertion-only streaming algorithm which achieves a constant approximation requires $\Omega(n/p)$ space. We note that our bound is specifically for dense graphs in line with our investigation into the separation between dense and sparse instances, and the hard instance we use is a family of dense graphs. In contrast, the hard instance of \cite{BKV2012densest} is a family of sparse graphs.

\paragraph{Dense CSPs.} For $\MaxCut$ and other CSPs, $O(\eps^{-2}n \log n)$-space streaming algorithms based on uniform subsampling (which preserve the value of every assignment up to an additive $\pm \eps$) have been studied in several previous works, including~\cite{Zel09,KK15,CGS+22-linear-space}.
Indeed, there is a moral reason why dense (maximization) CSPs are ``easier than'' sparse CSPs.
The goal is to satisfy as many constraints as possible, and having more constraints lets an algorithm narrow down the search space more efficiently (to either find a good assignment or conclude that none exists).
Thus, there are many models in which \emph{dense} CSPs are easier to solve than sparse CSPs, beyond the problem of outputting an approximately optimal assignment with a streaming algorithm which we study in this paper.
For instance, there are polynomial-time approximation schemes for many dense CSPs (see e.g.~\cite{Fer96,AKK99,MS08,FV22}), while approximating many sparse CSPs (up to arbitrarily good constant factors) is $\mathbf{NP}$-hard~\cite{ALM+98}.
Such schemes can be designed using simple subsampling procedures, like looking at the induced CSP on a constant-sized random subset of variables (e.g., \cite{AdKK03}); these ideas also extend to some LP/SDP relaxations \cite{BHHS11}.
Tradeoffs between running time and density are also known for CSPs in random~\cite{RRS17} and semirandom~\cite{GKM22} models.

\paragraph{Streaming CSPs.} There has been quite a lot of work on streaming algorithms for the $\MaxCut$ problem and other CSPs in recent years~\cite{KKS15,GVV17,KKSV17,BDV18,GT19,KK19,CGV20,AKSY20,AN21,CGSV21-boolean,CGS+22-monarchy,KP22,CGS+22-linear-space,SSSV23-random-ordering,Sin23-kand,SSSV23-dicut,SSV24-jour-version,CGSV24,SSSV25,FMW25,FMW26,ABFS25}. Note that these works are concerned with outputting an approximately optimal maximum cut \emph{value} versus an actual cut itself. 
For $\MaxCut$ (the value version) on sparse graphs, we have quite strong lower bounds: $\Omega_\eps(n)$ space is needed to distinguish between bipartite graphs ($\maxval{G}=1$) and graphs with $\maxval{G} \le \frac12+\eps$~\cite{KK19}. 
Also, \textcite{BDV18} showed a \emph{constant-space} algorithm for $(1\pm\eps)$-approximating the $\MaxCut$ value of a dense graph.
(In contrast, actually outputting an approximately optimal cut in this space is not possible according to our \Cref{thm:maxcut:dense hard}.)
\textcite{CGSV24} gave a $O(\log^{O(1)} n)$-versus-$\Omega(\sqrt n)$-space dichotomy theorem for sketching algorithms which approximate the value of families of $\MaxCSP$ problems.
\textcite{CGS+22-linear-space} built on their work and that of \cite{KK19} to get $\Omega(n)$-space lower bounds for some CSPs.
A dichotomy for multi-pass streaming was recently established by~\textcite{FMW26}.

\subsection{Road map}
In \Cref{sec:prelim} we give our preliminaries, and in \Cref{sec:TO} we give a technical overview of our results. Then we give our randomized lower bound for maximum cut for sparse graphs in \Cref{sec:randomized-LB-max-cut}, our maximum cut algorithm for dense graphs in \Cref{sec:upper-max-cut}, and our randomized lower bound for maximum cut in dense graphs in \Cref{sec:randomized-dense-LB-max-cut}. 
We give our deterministic lower bound for the maximum cut value in \Cref{sec:deter-LB-max-cut}. Finally, we present our results for densest subgraph in \Cref{sec:DS} and our extensions including to all CSPs, similarity and rarity in \Cref{sec:extensions}. 

\section{Preliminaries} \label{sec:prelim}
\paragraph{Terminology.} We consider undirected graphs $G = (V,E)$ where $V$ denotes the set of $n$ vertices and $E$ denotes the set of $m$ edges. We say $G$ is ``dense" if $m = \Omega(n^2)$ whereas $G$ is ``sparse" if $m = O(n)$. We may refer to graphs with $O(n/\eps^2)$ edges as ``sparse" (note that this is the same as $O(n)$ for constant $\eps$). 
We denote $N(u)$ as the set of adjacent vertices to vertex $u$ in a graph. In general, we boldface random variables. 

\paragraph{Sums of binomial coefficients.}

Let $\entropy(p) \coloneqq  -p\log_2 p - (1-p) \log_2 (1-p)$ denote the binary entropy function.
We recall the following standard lemma from combinatorics:
\begin{lemma}[Well-known, e.g., {\cite[Ch. 10, Cor. 9]{MS77}}]\label{lemma:prelim:hamming sum}
    Let $n \in \N$ and $0 \le k \le n/2$. Then
$\sum_{i=0}^{k} \binom{n}{i} \le 2^{n \cdot \entropy(k/n)}.$
\end{lemma}

\paragraph{The $\MaxCut$ problem.}
For an undirected $G = (V,E)$ and a function $x : V \to \{0,1\}$ representing a ``cut'', we define the \emph{cut value} of $x$ as the fraction of edges whose endpoints are assigned differently by $x$:
\[
\val{G}{x} \coloneqq \Pr_{\{\bu,\bv\}\sim \Unif{E}}[x(\bu) \ne x(\bv)].
\]
We define the \emph{value} of $G$ as
\[
\maxval{G} \coloneqq \max_{x : V \to \{0,1\}} \val{G}{x},
\]
the highest value of any cut.
$G$ is \emph{bipartite} if $\maxval{G} = 1$.
For every $\gamma \in [0,1]$, we define the set \[
\Good{\gamma}{G} \coloneqq \{ x : \val{G}{x} \ge (1-\gamma) \cdot \maxval{G} \}. \]
The \emph{$\eps$-approximate $\MaxCut$ problem} is to, given $G$, output a cut $x \in \Good{\eps}{G}$.

One can also consider the related ``value approximation'' problem where the goal is to output a real number $v \in [0,1]$ such that $v \in (1\pm\eps) \cdot \maxval{G}$.
Note that the value of a uniformly random cut $\bx : V \to \{0,1\}$ is $\frac12$ in expectation and therefore $\maxval{G} \ge \frac12$ for every graph $G$. 

\paragraph{One-way communication problems.}

We prove most of our lower bounds by considering the following simple communication problem, called $\IdProb$, in the one-way model.
\begin{problem}
    Let $\calF$ be a finite set.
    In the $\IdProb_\calF$ problem, Alice gets an input $A \in \calF$ and
    sends a single message $\calM$ to Bob.
    Bob, after arbitrary local computation, then needs to output $A$.
\end{problem}
The following information-theoretic bound is folklore:
\begin{lemma}
    In any one-way protocol for $\IdProb$ with  constant success probability,
    Alice must send Bob $\Omega(\log |\calF|)$ bits of information.
\end{lemma}

\section{Technical overview} \label{sec:TO}
Here we give an overview of our results. 

\paragraph{Randomized $\Omega(n \log(n \eps^2) / \eps^2)$ lower bound for maximum cut (\Cref{thm:main}).}
The goal of \Cref{thm:main} is to prove a lower bound against randomized streaming algorithms outputting an approximate maximum cut in a graph.

Given a fixed streaming algorithm and an input graph $G$, let $\mathtt{Compress}(G)$ denote the state of the streaming algorithm after processing the edges of $G$.
If the streaming algorithm uses space $s$, we can think of this as a (possibly randomized) mapping $\mathtt{Compress} : \{\text{graphs}\} \to \{0,1\}^s$.
Our fundamental idea is then to construct a large family $\calF$ of graphs on $\Theta(n)$ vertices ($\log|\calF| = \Omega(n \log(n\eps^2) / \eps^2)$)
and show that $\mathtt{Compress}(G)$ (typically) contains enough information about $G$ to uniquely identify it among all graphs in $\calF$.
Since $\mathtt{Compress}(G)$ has only $s$ bits, this is information-theoretically impossible unless $s \approx |\log(\calF)|$, allowing us to conclude the space lower bound.

Recall that $\mathtt{Compress}(G)$ encapsulates the state of the streaming algorithm after $G$'s edges are streamed.
Thus, knowledge of $\mathtt{Compress}(G)$ can help us to identify $G$ from among a large family of graphs $\calF$.
In particular, we can ``pick up where we left off'' in the stream;
that is, we can pick a graph $H$, add $H$'s edges to the stream, and then ask for an approximate maximum cut in $G \cup H$ (with good probability).
Modulo failure probability, we can now forget about the streaming model, and think instead about the following query problem:
when can $G \in \calF$ be identified given the ability to query an approximate maximum cut in $G \cup H$ for $H$ of our choosing?

Note that in the simplest case, we might set $H$ to the empty graph, and query an approximate maximum cut in $G$ itself.
However, this gives us just $n$ bits of information about $G$,\footnote{
    After all, a cut is just a string in $\{0,1\}^n$.}
which is not enough to identify $G$ in general (since $\log|\calF| \gg n$).
Instead, we will need to query with many different $H$'s.
But which $H$'s should we use to ``learn'' about $G$?

Our main conceptual insight is that, by carefully constructing the graph family $\calF$ and the query graphs $H$,
we can also compute certain approximate \emph{conditional} maximum cuts in $G$ itself.
By this, we mean the following.
If $G$ is a graph with vertex-set $V$, then a cut in $G$ is just a string in $\{\pm1\}^V$, i.e., a labeling of vertices by $\{\pm1\}$ values.
The standard maximum cut problem is to find the cut in $\{\pm1\}^V$ cutting the most edges in $G$.
In the ``conditional'' maximum cut problem, we are given, in addition to $G$, a subset $S \subseteq V$ of the vertices and a cut $x \in \{\pm1\}^S$ on those vertices,
and we want to find the cut $y \in \{\pm1\}^{V\setminus S}$ on the remaining vertices
such that the full cut $xy \in \{\pm1\}^V$ cuts as many edges as possible.
In other words, we want to maximize the number of cut edges not over the full cube $\{\pm1\}^V$, but over the subcube where the coordinates in $S$ are fixed to $x$.

Specifically, in our setup, the graphs $G \in \calF$ are certain bipartite graphs on a vertex set $L \sqcup R$, with $|L| = |R| = n$,
and for every partial cut $x \in \{\pm1\}^L$ on the left vertices, we hope to find an approximate conditional maximum cut $y \in \{\pm1\}^R$ on the right vertices.
We achieve this by defining a graph $H_x$ on vertex set $V' \supseteq V$ in such a way that when we query an approximate maximum cut $z \in \{\pm1\}^{V'}$ for $G \cup H_x$,
$z_L$ is forced to (mostly) equal $x$,
and $z_R$ correspondingly (mostly) recovers an approximate conditional maximum cut.
(Here $z_L$ and $z_R$ denote, respectively, $z$'s restriction onto the left and right vertices of $G$.)

Our technical work now splits into two parts:
\begin{enumerate}
\item We construct a very large family $\calF$ of bipartite graphs on vertex-set $L \sqcup R$ such that for every $G_1 \ne G_2 \in \calF$,
except with probability 1\% over the choice of the left-cut $\bx \in \{\pm1\}^L$, the set of approximate maximum right-cuts for $G_1$ and $G_2$ conditioned on $\bx$ are disjoint.
Specifically, to construct $\calF$, we sample a set of random bipartite graphs with right-regularity $k = \Theta(1/\epsilon^2)$,
and then filter out graphs which fail certain desired criteria.
\item We show how approximate maximum cut queries can be used to identify members of this family.
This step combines the standard Yao's minimax lemma / averaging ideas and a careful construction of graphs $H_x$.
\end{enumerate}

To conclude this subsection, we shed some additional light on Step (1), the construction of the large family of $k$-right-regular graphs $\calF$.
$\calF$ is required to satisfy two important properties:
firstly, the distribution of left-degrees of every graph in $\calF$ should be pseudorandom, in the sense that the total of squared left-degrees is $O(nk^2)$ and the total left-degree of vertices with left-degree $\Omega(k)$ is $O(n\sqrt{k})$.
Secondly, the pairwise overlap in the edge-sets of graphs in $\calF$ is small.

We then show that we can separate each pair of distinct graphs $G_1, G_2 \in \calF$. Fix such a pair and sample left assignment $\bx$ uniformly at random. 
We show that with high probability over $\bx$, there does not exist any right assignment that is simultaneously near-optimal for both conditional instances $(G_1, \bx)$ and $(G_2, \bx)$.
This is because fixing $x$ induces at each right vertex $v$ a signed imbalance which is the difference between the number of neighbors labeled $+1$ and $-1$ by $\bx$; any near-optimal right assignment must agree with these imbalances on most vertices.
Thus, if $G_1$ and $G_2$ shared a near-optimal right assignment under the same $x$, the difference in imbalance between the two has to be small.
On the other hand, the low-overlap property ensures that for most right vertices these imbalances are driven by essentially disjoint neighborhoods in $G_1$ and $G_2$, leading to a large expected inconsistency under random $\bx$.
The bounded degree conditions allow us to amplify this to high probability via concentration, ruling out shared near-optimal right assignments.
This separation implies that any algorithm which succeeds on a constant fraction of the instances must distinguish between a constant fraction of the pairs of graphs in $\calF$ yielding the desired lower bound. 

\paragraph{Randomized $O(n/\eps^2)$ upper bound for maximum cut.}
Our algorithm is very simple and is an application of the $F_0$ estimation algorithm from \cite{B2018optimal}. The algorithm of \cite{B2018optimal} estimates the number of distinct elements (so distinct edges here) in a stream up to a multiplicative $\eps$ approximation using $O(\log(\delta^{-1})/\eps^2 + \log n)$ space with probability $1-\delta$. For any cut $x$, take $\mathrm{cut}(x)$ to denote the set of edges that crosses $x$ in the complete graph. Inclusion-exclusion therefore gives us
\[
|G \cap \mathrm{cut}(x)| = |G| + |\mathrm{cut}(x)| - |G \cup \mathrm{cut}(x)|. 
\]
We can maintain $|G|$ exactly in $\log n$ bits and $|\mathrm{cut}(x)|$ is known offline. We can use the $F_0$ estimation algorithm to estimate $|G \cup \mathrm{cut}(x)|$ in the following way: run the algorithm on the input stream and then add the edges of $\mathrm{cut}(x)$ to the stream. Although the $F_0$ algorithm gives a multiplicative approximation to $|G \cup \mathrm{cut}(x)|$, this translates to $\eps n^2$ additive error on every cut value. Therefore, the density of the graph gives the desired approximation. 
Setting $\delta = 2^{-\Theta(n)}$ so that we can union bound over all cuts gives space $O(n/\eps^2)$. We also note that since we are using an $F_0$ sketch, we can handle the case where the same constraint is present in the stream multiple times. In this case, it is treated as if the constraint was only presented once.  

Given that we are working in the setting where $\delta$ is very small, the update time of the $F_0$ sketch of \cite{B2018optimal} is $\poly n$. It is unclear how to speed up the update time while preserving the exact space complexity (note that in contrast uniform sampling has fast update time but suboptimal space). So, we also present an alternate algorithm which has a slightly worse space dependence (by a $\log(1/\eps)$ factor) but has only amortized $\polylog n$ update time in \Cref{appen:altSamp}.

Our alternate algorithm goes as follows. 
The first key idea is to sample $n$ random edges before the stream, as opposed to the naive approach of sampling $n$ edges during the stream.
Then during the stream, we simply record which of the pre-sampled edges were actually present in the input graph with a simple bit array of size $n$ and use these edges as our sample. But how do we sample $n$ random edges before the stream in only $O(n)$ bits of space? Our second key idea is that we do not need truly i.i.d. uniform samples. 
So we construct a $d$-regular expander graph on vertex-set $\binom{n}{2}$ ($d$ depends only on $\eps$) and take a random walk of length $n$ in this graph starting from a random vertex. We can then use well-known concentration inequalities for random walks on expanders (see \cite{Vad12}) to conclude this gives a good estimator. Crucially, there is an $O(n+\log n)$-bit \emph{implicit representation} of the list of vertices on the walk (which correspond to possible edges in the graph), since all we need to store is the starting vertex and a ``next-step pointer'' saying ``go to my $i$-th neighbor'' for $i \in [d]$ for each vertex along the walk. We note that our techniques differ from \cite{B2018optimal} since we define the expander on all the possible edges of the input graph versus to generate random seeds. 

\paragraph{Randomized $\Omega(n)$ lower bound for maximum cut in dense graphs (\Cref{thm:maxcut:dense hard}).}
We note that since $|V|=n$, $n$ bits of space are required to \emph{write down} the output. However, our lower bound shows that $\Omega(n)$ bits of \emph{working} space (i.e., not counting the output space) is also required for this task.

To establish our lower bound, we design an exponentially large family of inputs $\calF$ where the sets of approximately optimal solutions are pairwise disjoint. We then show that one can leverage an approximate $\MaxCut$ streaming algorithm to correctly identify an input graph $G \in \calF$.
The set of inputs are complete bipartite graphs.

\paragraph{Deterministic lower bound for maximum cut value (\Cref{thm:deterministic-max-cut}).}
Our framework is similar to that of our randomized $\Omega(n \log(n\eps^2)/\eps^2)$ lower bound. We design a large family of graphs $\calF$ and show that $G \in \calF$ can be identified given the ability to query an approximate maximum cut value in $G \cup H$ for $H$ of our choosing. In the deterministic setting, the query access we assume is stronger than in the randomized setting: we may augment $G$ with any $H$ and obtain exact deterministic approximations to the maximum cut value of $G \cup H$ with no failure probability and therefore with an unbounded number of queries. 

We instantiate $\calF$ as the family of graphs formed by the union of $\Theta(1/\eps^2)$ perfect bipartite matchings on a vertex set of size $2n$. For every pair of distinct graphs $G_1, G_2 \in \calF$, we construct a deterministic distinguishing test $\mathcal T(G_1,G_2)$. Specifically, we show that for each such pair there exists a cut $x$ satisfying $|\val{G_1}{x}-\val{G_2}{x}| =\Omega(n/\eps).$ Using this cut, we define augmentation $H$ such that the maximum cut value of $G \cup H$ encodes $\val{G}{x}$. By thresholding the resulting approximate maximum cut value, we can deterministically distinguish whether the unknown graph $G$ is equal to $G_1$ or $G_2$ (with arbitrary behavior permitted otherwise). Running $\mathcal T(G_1,G_2)$ over all possible pairs of graphs in $\calF$ and selecting the graph consistent with the largest number of tests uniquely identifies $G$, yielding the desired space lower bound. 

\section{Randomized lower bound for maximum cut (\Cref{thm:main})} \label{sec:randomized-LB-max-cut}

We start by defining the notion of \emph{conditional} Max Cut,
where the input is a bipartite graph in which the left vertices have a fixed assignment
and the goal is to find the optimum assignment on the right vertices. We then give some lemmas which outline what conditions we want for our hard family of graphs. 

\begin{definition}[``Conditional'' Max Cut]
Let $G = (L \sqcup R, E)$ be a bipartite simple graph and $x \in \{\pm1\}^L$ a fixed assignment to the left vertices.
The \emph{value} of an assignment $y \in \{\pm1\}^R$ to the right vertices is \[
\Cval{G}{x}{y} \coloneqq \frac12 \sum_{(u,v) \in E} (1 - x_u y_v) = \frac{|E|}2 - \frac12 \sum_{(u,v) \in E} x_u y_v. \]
The optimum value over all assignments is \[
\Copt{G}{x} \coloneqq \max_{y\in\{\pm1\}^R} \Cval{G}{x}{y}. \qedhere \]
\end{definition}
Given $G$ and $x$, we define the \emph{discrepancy} of a vertex $v \in R$ as
\begin{equation}\label{eq:discrepancy}
\Disc{G}{x}{v} \coloneqq \sum_{u \in \Nbr{G}{v}} x_u.
\end{equation}
This immediately gives, for a fixed right-assignment $y \in \{\pm1\}^R$, the equation
\begin{equation}\label{eq:Cval}
\Cval{G}{x}{y} = \frac{|E|}{2} - \frac{1}{2} \sum_{v \in R} y_v \cdot \Disc{G}{x}{v}.
\end{equation}
Maximizing this over all right-assignments gives the equation
\begin{equation}\label{eq:Copt}
\Copt{G}{x} = \frac{|E|}{2} + \frac{1}{2} \sum_{v \in R} |\Disc{G}{x}{v}|.
\end{equation}
We can then define the \emph{loss} of a particular right-assignment $y$ as:
\begin{equation}\label{eq:Closs}
    \Closs{G}{x}{y} \coloneqq \Copt{G}{x} - \Cval{G}{x}{y} = \frac{1}{2} \sum_{v \in R} \parens*{  |\Disc{G}{x}{v}| + y_v \cdot \Disc{G}{x}{v} }.
\end{equation}
Finally, for a goodness threshold $\tau \in \R_{> 0}$, we define the set of good right-assignments for $x$ as:
\begin{equation}\label{eq:Cgood}
\Cgood{G}{x}{\tau} \coloneqq \braces*{ y \in \{\pm1\}^R : \Closs{G}{x}{y} \le \tau }.
\end{equation}

For $a,b \in \R$, define $\slack(a,b) \coloneqq |a| + |b| - |a + b|$.
This quantity is the slack in the triangle inequality applied to $a$ and $b$.
In particular, it is always nonnegative.
It vanishes iff $a$ and $b$ have the same sign (or either is zero),
and otherwise equals $\min(|a|,|b|)$.

Let $\vecG = (G_1,G_2)$ denote an ordered pair of bipartite graphs on the same vertex-set $L \sqcup R$.
For a fixed left-assignment $x \in \{\pm1\}^L$ and right-vertex $v \in R$, we define the \emph{advantage}
\begin{equation}\label{eq:adv}
\adv{\vecG}{x}{v} \coloneqq \slack(\Disc{G_1}{x}{v}, \Disc{G_2}{x}{v}).
\end{equation}
Our first lemma states that if two graphs share a good right-assignment, then the total advantage must be small.
\begin{lemma}\label{lem:nsa implies small adv}
    For every pair of bipartite graphs $\vecG = (G_1,G_2)$ and left-assignment $x \in \{\pm1\}^L$,
    if $\Cgood{G_1}{x}{\tau} \cap \Cgood{G_2}{x}{\tau} \neq \emptyset$,
    then $\sum_{v \in R} \adv{\vecG}{x}{v} \le 4\tau$.
\end{lemma}

\begin{proof}
    If $y \in \Cgood{G_1}{x}{\tau} \cap \Cgood{G_2}{x}{\tau}$,
then by definition, we have $\Closs{G_1}{x}{y} + \Closs{G_2}{x}{y} \le 2\tau$.
But all $y$'s must incur total loss at least:
\begin{align*}
&\min_{y\in \{\pm1\}^R} \parens*{ \Closs{G_1}{x}{y} + \Closs{G_2}{x}{y} } \\
&= \frac12 \min_{y\in \{\pm1\}^R} \sum_{v \in R} \parens*{ |\Disc{G_1}{x}{v}| + |\Disc{G_2}{x}{v}| + y_v (\Disc{G_1}{x}{v} + \Disc{G_2}{x}{v})} \\
&= \frac12 \sum_{v \in R} \parens*{ |\Disc{G_1}{x}{v}| + |\Disc{G_2}{x}{v}| - |\Disc{G_1}{x}{v} + \Disc{G_2}{x}{v}| } \\
&= \frac{1}{2} \sum_{v \in R} \adv{\vecG}{x}{v}. \qedhere \end{align*}
\end{proof}

Now, we establish conditions under which the advantage is large.
For $v \in R$ and $x \in \{\pm1\}^L$, we decompose
\begin{align*}
    \xA{\vecG}{x}{v} &\coloneqq \sum_{u \in \Nbr{G_1}{v} \setminus \Nbr{G_2}{v}} x_u, \\
    \xB{\vecG}{x}{v} &\coloneqq \sum_{u \in \Nbr{G_2}{v} \setminus \Nbr{G_1}{v}} x_u, \\
    \xCom{\vecG}{x}{v} &\coloneqq \sum_{u \in \Nbr{G_1}{v} \cap \Nbr{G_2}{v}} x_u,
\end{align*}
We then give the following lower bound on $\adv{\vecG}{x}{v}$:
\begin{lemma}\label{lem:adv lower bound}
    For every pair of bipartite graphs $\vecG = (G_1,G_2)$, left-assignment $x \in \{\pm1\}^L$, and right-vertex $v \in R$: \[
    \adv{\vecG}{x}{v} \ge \slack \parens*{ \xA{\vecG}{x}{v}, \xB{\vecG}{x}{v}} - 4|\xCom{\vecG}{x}{v}|. \]
\end{lemma}

To prove this, we use the following fact:

\begin{lemma}[Lipschitzness of the slack function] \label{lem:lipschitz}
The function $\slack : \R \times \R \to \R$ is 2-Lipschitz with respect to the $1$-norm:
$|\slack(a', b') - \slack(a, b)| \le 2(|a' - a| + |b' - b|)$.
\end{lemma}
\begin{proof}
    Firstly, the absolute value function itself is $1$-Lipschitz:
    We have $a' = a + (a'-a) \implies |a'| \le |a| + |a'-a| \implies |a'|-|a| \le |a'-a|$.
    The same holds swapping $a$ and $a'$, hence $||a'|-|a|| \le |a'-a|$.
    Now, we check: 
    \begin{align*}
    |\slack(a',b') - \slack(a,b)| &= \abs*{ |a'|+|b'|-|a'+b'| - |a| - |b| + |a + b| } \\
    &\le \abs*{ |a'|-|a| } + \abs*{ |b'|-|b| } + \abs*{ |a'+b'|-|a+b|} \tag{tri. ineq.} \\
    &\le |a'-a| + |b'-b| + |a'+b'- (a+b)| \tag{$1$-Lipschitzness of $|\cdot|$} \\
    &\le |a'-a| + |b'-b| + |a'-a| + |b'-b| \tag{tri. ineq.} \\
    &= 2(|a'-a|+|b'-b|). \qedhere
    \end{align*}
\end{proof}

\begin{proof}[Proof of \Cref{lem:adv lower bound}]
Since $\Nbr{G_1}{v} = (\Nbr{G_1}{v} \setminus \Nbr{G_2}{v}) \sqcup (\Nbr{G_1}{v} \cap \Nbr{G_2}{v})$, we have $\Disc{G_1}{x}{v} = \xA{\vecG}{x}{v} + \xCom{\vecG}{x}{v}$
and similarly for $\Disc{G_2}{x}{v}$.
Hence, using the triangle inequality, we have:
\begin{multline*}
\adv{\vecG}{x}{v} = \slack \parens*{ \xA{\vecG}{x}{v} + \xCom{\vecG}{x}{v},\xB{\vecG}{x}{v} + \xCom{\vecG}{x}{v} } \\
\ge \slack \parens*{ \xA{\vecG}{x}{v}, \xB{\vecG}{x}{v} }
- \abs*{ \slack \parens*{ \xA{\vecG}{x}{v} + \xCom{\vecG}{x}{v},\xB{\vecG}{x}{v} + \xCom{\vecG}{x}{v}} - \slack \parens*{ \xA{\vecG}{x}{v}, \xB{\vecG}{x}{v}}}
\end{multline*}
and we can then apply  \Cref{lem:lipschitz}.
\end{proof}

\subsection{Probabilistic construction of a hard family of graphs}

We construct the hard family $\calF$ of graphs using the so-called ``deletion method'':
we sample a large family of random graphs, delete all graphs which violate desired criteria,
and show that the resulting family is still large because violations are improbable.

Let $\GRR{n}{k}$ be the space of $k$-right-regular bipartite simple graphs, generated by having each $v \in R$ independently choose a uniform random $k$-subset of $L$ as its neighbors.

For a bipartite graph $G = (L \sqcup R, E)$ and a threshold $d$, we define
\begin{align*}
\Tail{G}{d} \coloneqq \{u \in L : \deg{G}{u} > d\} && \text{and} &&
\TailDeg{G}{d} \coloneqq \sum_{u \in \Tail{G}{d}} \deg{G}{u}.
\end{align*}

\newcommand{\etaOverlap}{\eta_{\mathrm{near}}}
\newcommand{\etaTail}{\eta_{\mathrm{tail}}}
\newcommand{\Cdeg}{C_{\mathrm{deg}}}
\newcommand{\Csize}{C_{\mathrm{size}}}
\newcommand{\Cstretch}{C_{\mathrm{stretch}}}
\newcommand{\pnear}{p_{\mathrm{near}}}

\begin{lemma}[Existence of hard family] \label{lem:hard_family}
For every $\etaOverlap,\etaTail > 0$ and $\Cdeg > 1$, there exist $\Csize, \Cstretch > 0$ such that the following holds: 
for every sufficiently large $k \in \N$ and $n \ge \Cstretch k \in \N$,
there exists a family $\calF \subset \GRR{n}{k}$ of simple bipartite graphs satisfying:
\begin{enumerate}
    \item \emph{Right-regularity}:
    for every $G \in \calF$ and $v \in R$, it holds that $\deg{G}{v} = k$.
    \item \emph{Bounded sum of squared left-degrees:}
    for every $G \in \calF$, it holds that $\sum_{u \in L} \deg{G}{u}^2 \le \Cdeg n k^2$.
    \item \emph{Light left-degree tail:}
    for every $G \in \calF$, it holds that $\TailDeg{G}{2k} \le \etaTail n\sqrt{k}$.
    \item \emph{Low overlap:} for every $G_1 \ne G_2 \in \calF$, it holds that $|E(G_1) \cap E(G_2)| \le \etaOverlap nk$.
    \item \emph{Large size:} $\log |\calF| \ge \Csize nk \log(n/k)$.
\end{enumerate}
\end{lemma}

Note that the first three properties impose requirements on each graph in the family $\calF$ individually;
the fourth property involves pairs of graphs,
and the fifth property involves the size of the entire family.

\begin{proof}
By definition, any $\bG \sim \GRR{n}{k}$ deterministically satisfies Property 1.

Pick some fixed $p_2 \in (1/\Cdeg,1)$.
For Property 2, sampling $\bG \sim \GRR{n}{k}$, the degree $\deg{\bG}{u}$ of a left vertex $u$ is the sum of $n$ independent indicator variables (one for each $v \in R$ choosing $u$) with expectation $k/n$.
Thus $\deg{\bG}{u} \sim \text{Binomial}(n, k/n)$.
We have $\Exp[\deg{\bG}{u}] = k$ and $\Var(\deg{\bG}{u}) \le k$.
Therefore, $\Exp[\deg{\bG}{u}^2] = \Var(\deg{\bG}{u}) + \Exp[\deg{\bG}{u}]^2 \le k + k^2$. 
By linearity of expectation, $\Exp\bracks*{ \sum_{u \in L} \deg{\bG}{u}^2} \le n(k + k^2) \le p_2 \Cdeg nk^2$ (taking $k$ large enough).
By Markov's inequality, the probability that this sum exceeds $\Cdeg nk^2$ is at most $(p_2\Cdeg)/\Cdeg = p_2 < 1$.

For Property 3, sampling $\bG \sim \GRR{n}{k}$, the expected total degree of the tail vertices is:
\begin{align*}
    \Exp[\TailDeg{\bG}{2k}] = \sum_{u \in L} \Exp\bracks*{ \deg{\bG}{u} \mathbf{1}_{\deg{\bG}{u} > 2k}} 
    &= n \sum_{j=2k+1}^n j \binom{n}{j} (k/n)^j (1-k/n)^{n-j} \\
    &= n k \sum_{i=2k}^{n-1} \binom{n-1}{i} (k/n)^i (1-k/n)^{n-1-i} \\
    &= n k \cdot \Pr(\text{Binomial}(n-1, k/n) \ge 2k).
\end{align*}
By standard multiplicative Chernoff bounds (since the mean is strictly $< k$),
this upper tail probability is bounded by $\le e^{-k/3}$, hence $\Exp[\TailDeg{\bG}{2k}] \le nk \cdot e^{-k/3}$.
Hence by Markov's inequality, $\Pr[\TailDeg{\bG}{2k} > \etaTail n \sqrt{k}] \le (nk \cdot e^{-k/3})/(\etaTail n \sqrt{k}) = \etaTail^{-1} \sqrt{k} e^{-k/3}$.
Taking $k$ sufficiently large, this is smaller than $0.49 (1-p_2)$.

For Property 4, consider sampling $\bG_1,\bG_2 \sim \GRR{n}{k}$ independently, and define the overlap $\bX := |E(\bG_1) \cap E(\bG_2)| = \sum_{v \in R} |\Nbr{\bG_1}{v} \cap \Nbr{\bG_2}{v}|$.
Because each $v \in R$ chooses its $k$ neighbors independently, $\bX$ is the sum of $n$ independent 
hypergeometric
random variables.
Let $\pnear$ denote the probability that $\bX$ exceeds $\etaOverlap n k$.
Since $\Exp[\bX] = n \cdot (k^2/n) = k^2$,
using Hoeffding's extension of Chernoff bounds for sums of hypergeometrics,
we get
\begin{multline*}
\pnear \le (e k^2 / (\etaOverlap nk))^{(\etaOverlap nk)} = \exp(-\etaOverlap nk \ln (\etaOverlap n/(ek))) 
\leq \exp(-\tfrac12 \etaOverlap nk \ln(n/k)),
\end{multline*}
where the final inequality uses the assumption on $n$ and takes $\Cstretch \gg e/\etaOverlap$.

We now apply the deletion method: sample $N \coloneqq \lfloor 0.98(1-p_2)/(\pnear) \rfloor$ graphs independently.
The expected number of graphs violating Property 2 or 3 is $\le (p_2 + 0.49(1-p_2))N < N$.
The expected number of pairs violating Property 4 is $\binom{N}{2} \pnear < \frac12 (\pnear N) N < 0.49(1-p_2) N$.
Thus, the expected total number of violations is at most $(p_2 + 0.98(1-p_2))N$.
Hence, there exists a realization with at most $(p_2 + 0.98(1-p_2))N$ violations.
Removing one graph from each violation leaves a simple family $\calF$ satisfying all properties with size $|\calF| \ge 0.02 (1-p_2) N$. 
Taking the logarithm yields Property 5. \qedhere
\end{proof}

\subsection{Separating conditional Max Cut instances}\label{sec:params}
Now we show that distinct graphs in the hard family admit
incompatible sets of near-optimal right-assignments once the left assignment
$x$ is fixed at random. In particular, for two different graphs $G_1,G_2$,
with high probability over the choice of $x$, there does not exist any
right-assignment that is simultaneously near-optimal for both conditional
Max Cut instances $(G_1,x)$ and $(G_2,x)$. We will later show how to lift this to a lower bound for the standard (unconditional) Max Cut
problem.

Define the density parameter $\Cden \coloneqq 10^6$, and a corresponding regularity parameter $k \coloneqq \floor{ 1/(10^5 \Cden \epsilon)^2 }$.
Let $\epsilon_0 > 0$ be a sufficiently small absolute constant such that $k$ is large enough to invoke the prior lemma for all $\epsilon \le \epsilon_0$.
The slack threshold is defined as $\tau \coloneqq \frac{n\sqrt{k}}{2000}$. 
We will invoke \Cref{lem:hard_family} with $\etaTail \coloneqq \frac1{8000}$, $\etaOverlap \coloneqq 10^{-8}$, and $\Cdeg \coloneqq 10$.


\begin{lemma}[Slack vs. minimum magnitude] \label{lem:equiv}
Let $\bX, \bY$ be independent, symmetrically distributed random variables.
Then $\Exp[\slack(\bX,\bY)] = \Exp[\min(|\bX|, |\bY|)]$.
\end{lemma}
\begin{proof}
Observe that for any real numbers $a,b$, we have $|a+b| + |a-b| = 2\max(|a|, |b|)$.
Thus, \[
\slack(a,b) + \slack(a,-b) = 2|a| + 2|b| - (|a+b| + |a-b|) = 2|a| + 2|b| - 2\max(|a|, |b|) = 2\min(|a|, |b|). \]

Since $\bY$ is symmetric, $-\bY$ has the same distribution as $\bY$.
Because $\bX$ and $\bY$ are independent, the joint distribution of $(\bX, -\bY)$ is identical to $(\bX, \bY)$. Therefore, $\Exp[\slack(\bX, -\bY)] = \Exp[\slack(\bX, \bY)]$.
Taking the expectation of the sum gives $2\Exp[\slack(\bX, \bY)] = 2\Exp[\min(|\bX|, |\bY|)]$, which completes the proof.
\end{proof}

\begin{lemma}[Minimum magnitude of Rademacher sums] \label{lem:rademacher}
Let $\bX,\bY$ be independent Rademacher sums of lengths $m_X, m_Y \ge m \ge 1$. Let $\bW \coloneqq \min(|\bX|,|\bY|)$.
Then $\Exp[\bW] \ge \Csep \sqrt{m}$, where $\Csep \coloneqq 9/512 \approx 0.0175$.
\end{lemma}
\begin{proof}
We use the Paley--Zygmund inequality on $\bX^2$.
Since $\bX$ is a sum of $m_X$ independent Rademacher variables, $\Exp[\bX^2] = m_X$.
The fourth moment expands to $\Exp[\bX^4] = m_X + 3m_X(m_X-1) = 3m_X^2 - 2m_X < 3m_X^2$.
By Paley--Zygmund, $\Pr(\bX^2 > m_X/4) \ge (1-\frac14)^2 \frac{\Exp[\bX^2]^2}{\Exp[\bX^4]} > \frac9{16} \frac{m_X^2}{3m_X^2} = \frac{3}{16}$.
Since $m_X \ge m$, $\Pr(|\bX| > \sqrt{m}/2) \ge \Pr(|\bX| > \sqrt{m_X}/2) \ge 3/16$.
Since $\bX, \bY$ are independent, $\Pr(\bW > \sqrt{m}/2) = \Pr(|\bX| > \sqrt{m}/2)\Pr(|\bY| > \sqrt{m}/2) \ge (3/16)^2 = 9/256$.
Because $\bW \ge 0$, we have $\Exp[\bW] \ge (\sqrt{m}/2) \cdot (9/256) = \frac{9}{512}\sqrt{m} = \Csep\sqrt{m}$.
\end{proof}

\begin{lemma}[Typically no shared good assignments] \label{lem:nsa}
With the parameters we defined, let $\calF$ be as in \Cref{lem:hard_family} and let $G_1, G_2 \in \calF$ be distinct.
Then \[
\Pr_{\bx \in \{\pm1\}^L} \bracks*{ \Cgood{G_1}{\bx}{\tau} \cap \Cgood{G_2}{\bx}{\tau} \neq \emptyset } \le 0.01. \]
\end{lemma}
\begin{proof}
Let $F(x) \coloneqq \sum_{v \in R} \adv{\vecG}{x}{v}$.
By \Cref{lem:adv lower bound}, we have $\Cgood{G_1}{x}{\tau} \cap \Cgood{G_2}{x}{\tau} \neq \emptyset \implies F(x) \le 4\tau \eqqcolon T^*$.
We now upper bound the probability over $\bx \in \{\pm1\}^L$ that $F(\bx) \le T^*$.
Note $T^* = \frac{4 n \sqrt{k}}{2000} = 0.002 n\sqrt{k}$.

First, we lower bound the expectation of $F(\bx)$.
Let $\rho \coloneqq 10^{-6}$
and define the set of typical vertices $\Vtyp \coloneqq \{v \in R : |\Nbr{G_1}{v} \cap \Nbr{G_2}{v}| \le \rho k\}$.
By the low overlap property, $\Exp_\bv [|\Nbr{G_1}{\bv} \cap \Nbr{G_2}{\bv}|] \le \etaOverlap k$.
By the setting of $\etaOverlap$ and Markov's inequality, we conclude  $|\Vtyp| \ge 0.99n$. 

Fix $v \in \Vtyp$.
Because the graphs are simple (no parallel edges),
when $\bx \sim \{\pm1\}^L$, $\xCom{\vecG}{\bx}{v}$ is a sum of $|\Nbr{G_1}{v} \cap \Nbr{G_2}{v}| \le \rho k$ independent Rademacher random variables.
Thus $\Exp[|\xCom{\vecG}{\bx}{v}|] \le \sqrt{\Var(\xCom{\vecG}{\bx}{v})} = \sqrt{|\Nbr{G_1}{v} \cap \Nbr{G_2}{v}|} \le \sqrt{\rho k}$.
Further, because $G_1, G_2$ are $k$-right regular, $\xA{\vecG}{\bx}{v}$, $\xB{\vecG}{\bx}{v}$ are independent sums of $m \coloneqq k - |\Nbr{G_1}{v} \cap \Nbr{G_2}{v}| \ge (1-\rho)k$ independent Rademacher random variables. 
Since $\xA{\vecG}{\bx}{v}, \xB{\vecG}{\bx}{v}$ are independent symmetric variables, by \Cref{lem:rademacher,lem:equiv}, \[ \Exp[\slack(\xA{\vecG}{\bx}{v}, \xB{\vecG}{\bx}{v})] = \Exp[\min(|\xA{\vecG}{\bx}{v}|, |\xB{\vecG}{\bx}{v}|)] \ge \Csep \sqrt{m}. \]
We finally have, sampling $\bx \in \{\pm1\}^L$:
\begin{multline*}
\Exp[\adv{\vecG}{\bx}{v}] \ge \Exp[\slack(\xA{\vecG}{\bx}{v}, \xB{\vecG}{\bx}{v})] - 4\Exp[|\xCom{\vecG}{\bx}{v}|] \\
\ge \Csep\sqrt{m} - 4\sqrt{\rho k} \ge \sqrt{k}(\Csep\sqrt{1-\rho} - 4\sqrt{\rho}).
\end{multline*}
Substituting our constants ($\rho = 10^{-6}, \Csep \approx 0.0175$), we deduce $\Exp[\adv{\vecG}{\bx}{v}] \ge 0.0135\sqrt{k}$.
The total expectation therefore satisfies $\Exp[F(\bx)] \ge 0.99n \cdot 0.0135\sqrt{k} \ge 0.0133 n\sqrt{k}$.
The required deviation for a failure is $t \coloneqq \Exp[F(\bx)] - T^* \ge 0.011 n\sqrt{k}$.

Now, we apply McDiarmid's bounded differences inequality to the function $F : \{\pm1\}^L \to \R$ applied to the independent Rademacher variables $(\bx_u \in \{\pm1\})_{u \in L}$.
Given a fixed $x \in \{\pm1\}^L$,
flipping $x_u$ changes $\Disc{G_i}{\bx}{v}$ by $2$ if $v \in \Nbr{G_i}{u}$ and $0$ otherwise.
Hence, by \Cref{lem:lipschitz}, the maximum change in $F(x)$ from flipping $x_u$ is strictly bounded by $c_u \coloneqq 4(\deg{G_1}{u} + \deg{G_2}{u})$.

The variance proxy $\Sigma$ is elegantly bounded via the squared $2$-norm of degrees property (\Cref{lem:hard_family}, Item 2), applying the algebraic inequality $(a+b)^2 \le 2a^2 + 2b^2$: \[
\Sigma = \sum_{u \in L} c_u^2 \le 16 \sum_u (\deg{G_1}{u} + \deg{G_2}{u})^2 \le 32 \sum_u \parens*{ \deg{G_1}{u}^2 + \deg{G_2}{u}^2} \le 640 n k^2. \]
Applying McDiarmid's inequality: \[
\Pr [F(\bx) \le T^*] \le \exp\parens*{  - \frac{2t^2}{\Sigma} } \le \exp\parens*{  - \frac{2(0.011 n\sqrt{k})^2}{640 n k^2} } = \exp\parens*{  - \Omega\parens*{ \frac{n}{k}} }. \]
Since $n \geq \Cstretch k$, taking $\Cstretch$ sufficiently large this failure probability is at most $0.01$. 
\end{proof}

\subsection{Reducing conditional Max Cut to Max Cut}

\begin{definition}[Gadget graph]
    Let $L, R$ be disjoint vertex sets.
    For $x \in \{\pm1\}^L$, we define the following ``gadget'' graph $H_x$:
    The vertex set is $L \sqcup R \sqcup \{s_{+1},s_{-1}\}$, where $s_{+1}, s_{-1}$ are two new ``sink'' vertices;
    there is an edge of weight $\Wss \coloneqq 10^5 n k$ between $s_{+1}$ and $s_{-1}$, and, for each $u \in L$, an edge between $u$ and $s_{-x_u}$ with weight $\Wsl \coloneqq 2k$.
\end{definition}

Note that the graphs $G$ and $H_x$ are individually bipartite, but their union $G \cup H_x$ need not be.

\begin{lemma}\label{lem:cond to plain}
    Let the parameters be as in \Cref{sec:params}.
    Suppose $G \in \calF$ and $x \in \{\pm1\}^L$.
    Let $z \in \{\pm1\}^{L \sqcup R \sqcup \{s_{\pm1}\}}$ be an assignment to the graph $G \cup H_x$.
    If $\val{G \cup H_x}{z} \ge (1-\epsilon)\cdot \opt{G \cup H_x}$,
    then $z_R \in \Cgood{G}{x}{\tau}$,
    where $z_R$ denotes the induced assignment on the vertices $R$.
\end{lemma}

\begin{proof}
    Let \[
    \ell \coloneqq \opt{G \cup H_x} - \val{G \cup H_x}{z}, \]
    denote the loss of the assignment $z$, so that we have assumed $\ell \le \epsilon \cdot \opt{G \cup H_x}$.
    We have $\opt{G \cup H_x} \le \Cden nk$ (since the latter upper-bounds the total weight in $G \cup H_x$), hence $\ell \le \epsilon \cdot \Cden nk$.
    By the setting of $\Cden$, we have $k = \lfloor 1/(10^{10} \epsilon^2 \Cden^2) \rfloor \le 1/(10^{10} \epsilon^2 \Cden^2)$,
    and hence $\epsilon \sqrt{k} \le 1/(10^5 \Cden)$.
    This guarantees 
    \begin{equation}\label{eq:lambda upper bound}
    \ell \le \frac{n\sqrt{k}}{10000}.
    \end{equation}

    Consider the assignment $z^* \in \{\pm1\}^{L \sqcup R \sqcup \{s_{\pm1}\}}$
    defined by $z^*(s_{+1}) \coloneqq +1$, $z^*(s_{-1}) \coloneqq -1$,
    $z^*(u) \coloneqq x_u$ for $u \in L$, and $z^*(v) \coloneqq -\sign(\Disc{G}{x}{v})$ for $v \in R$.
    This assignment cuts the sink-sink edge and all sink-left vertex edges.
    Hence
    \begin{equation}\label{eq:opt lower bound}
    \opt{G\cup H_x} \ge \val{G\cup H_x}{z^*} = \Wss + n \Wsl + \Copt{G}{x}.
    \end{equation}
    
    If $z(s_{+1}) = z(s_{-1})$, then the sink-sink edge is not cut, so $\val{G\cup H_x}{z} \le n \Wsl + |E(G)|$.
    \Cref{eq:opt lower bound} then gives $\ell \ge \Wss - (|E(G)| - \Copt{G}{x}) \ge 10^5 nk - nk$, contradicting \Cref{eq:lambda upper bound}.
    Hence, $z(s_{+1}) \ne z(s_{-1})$.
    By symmetry, we assume that $z(s_{+1}) = +1$ and $z(s_{-1}) = -1$.

    Let $\Delta \coloneqq \{u \in L : z_L(u) \neq x_u\}$.
    We claim that 
    \begin{equation}\label{eq:lambda lower bound}
    \ell \ge \sum_{u \in \Delta} \parens*{  2k - \deg{G}{u} } + \Closs{G}{x}{z_R}. \end{equation}
    Indeed, consider the ``intermediate'' assignment $z' \in \{\pm1\}^{L \sqcup R \sqcup \{s_{\pm1}\}}$
    which matches $z$ on the right vertices and $z^*$ (and therefore $x$) on the left vertices.
    (That is, $z'(s_{\pm1}) = \pm 1$, $z'(u) = x_u$ for $u \in L$, and $z'(v) = z_v$ for $v \in R$.)
    We have
    \begin{multline*}
    \ell = \opt{G \cup H_x} - \val{G \cup H_x}{z}
    \ge \val{G \cup H_x}{z^*} - \val{G \cup H_x}{z} \\
    = \underbrace{(\val{G \cup H_x}{z^*} - \val{G \cup H_x}{z'})}_{\eqqcolon \Lsl} + \underbrace{(\val{G \cup H_x}{z'} - \val{G \cup H_x}{z})}_{\eqqcolon \Llr}.
    \end{multline*}
    For the first term, $\Lsl$, we have $\val{G \cup H_x}{z^*} = \Wss + n\Wsl + \Copt{G}{x}$,
    while $\val{G \cup H_x}{z'} = \Wss + n\Wsl + \Cval{G}{x}{z_R}$.
    (In other words, both assignments satisfy all edges incident to the sinks.
    Here, $z_R$ denotes the assignment $z$ restricted to the right vertices.)
    Hence, the difference is $\Lsl = \Closs{G}{x}{z_R}$.
    Meanwhile, we can write the second term as
    \begin{multline*}
    \Llr = \sum_{u \in \Delta} \Big((\text{weight of edges incident to }u\text{ satisfied by }z'\text{ but not }z) \\
    - (\text{weight of edges incident to }u\text{ satisfied by }z\text{ but not }z') \Big).
    \end{multline*}
    For every $u \in \Delta$, the first term is at least $\Wsl = 2k$, while the second is pessimistically at most $\deg{G}{u}$.
    Hence,
    \begin{equation}\label{eq:lambda lower bound 2}
    \ell \ge \Closs{G}{x}{z_R} + \sum_{u \in \Delta} (2k-\deg{G}{u}).
    \end{equation}

    Finally, recall the definition of the ``tail vertices'' $\Tail{G}{2k} = \{u : \deg{G}{u} > 2k\}$.
    We split the sum in \Cref{eq:lambda lower bound 2} over $u \in \Delta$ into $u \in \Tail{G}{2k}$ and $u \not\in \Tail{G}{2k}$;
    in the latter case, we have immediately that $(2k - \deg{G}{u}) \ge 0$,
    while in the former, we use the pessimistic lower bound $2k - \deg{G}{u} \ge -\deg{G}{u}$.
    Therefore: \[
    \ell \ge \Closs{G}{x}{z_R}-\sum_{u \in \Tail{G}{2k}} \deg{G}{u} = \Closs{G}{x}{z_R} -\TailDeg{G}{2k}. \]
    By the light tail property and our setting of $\etaTail$, $\TailDeg{G}{2k} \le \frac{n\sqrt{k}}{8000}$.
    Rearranging yields: \[
    \Closs{G}{x}{z_R} \le \ell + \TailDeg{G}{2k} \le \frac{n\sqrt{k}}{10000} + \frac{n\sqrt{k}}{8000} < \frac{n\sqrt{k}}{2000} = \tau. \qedhere \]
\end{proof}

\subsection{Streaming reduction}

We prove \Cref{thm:main} by reducing from the $\IdProb$ problem on $\calF$, the hard family of graphs from \Cref{lem:hard_family}.

\begin{proof}[Proof of \Cref{thm:main}]

The hard distribution $\calD$ is as follows: pick $\bG$ uniformly from $\calF$, and $\bx$ uniformly from $\{\pm1\}^L$.
The stream is $\bG \circ H_\bx$, where $\circ$ denotes ordered concatenation (and $\bG$ and $H_\bx$ themselves may be in arbitrary order).

Assume $\calA$ uses space $S$.
By Yao's minimax principle and averaging, there is a deterministic algorithm $\calA'$ and subset $\calF' \subseteq \calF$ with $|\calF'| \ge \frac12|\calF|$ such that for every $G \in \calF'$, \[
\Pr_{\bx \in \{\pm1\}^L}[\val{G \cup H_\bx}{\calA'(G \circ H_\bx)} \ge (1-\epsilon) \cdot \opt{G \cup H_\bx}] \ge 0.51. \]
This implies, by \Cref{lem:cond to plain} and taking the complement, that
\[
\Pr_{\bx \in \{\pm1\}^L}[\calA'_R(G \circ H_\bx) \not\in \Cgood{G}{x}{\tau}] \le 0.49, \]
where $\calA'_R(G \circ H_x)$ denotes the projection of the assignment $\calA(G \circ H_x)$ to only the vertices $R$.

Assume for contradiction that $S < \log_2(|\calF'|)$.
By the pigeonhole principle, there exist distinct $G_1, G_2 \in \calF'$ generating identical memory states before $H_x$ arrives.
Hence, for every $x \in \{\pm1\}^L$, $\calA'(G_1,x) = \calA'(G_2,x)$ and in particular $\calA'_R(G_1,x) = \calA'_R(G_2,x)$.
So by a union bound, \[
\Pr_{\bx \in \{\pm1\}^L}[\calA'_R(G_1 \circ H_\bx) \not\in \Cgood{G_1}{\bx}{\tau} \vee \calA'_R(G_1 \circ H_\bx) \not\in \Cgood{G_2}{\bx}{\tau}] \le 0.98. \]
Hence, $\Cgood{G_1}{\bx}{\tau} \cap \Cgood{G_2}{\bx}{\tau}$ is non-empty with probability $ \ge 0.02$.
This directly contradicts \Cref{lem:nsa}, which bounds this collision probability strictly $\le 0.01$.

Thus, $S \ge \log(|\calF'|) = \Omega(nk \log(n/k))$.
Substituting $k = \Theta(\epsilon^{-2})$ yields the final bound $S = \Omega(\epsilon^{-2} n \log(n\epsilon^2))$.
\end{proof}

\section{Upper bound for maximum cut in dense graphs} \label{sec:upper-max-cut}
\begin{theorem}\label{thm:maxcut:dense alg}
There exists an insertion-only randomized streaming algorithm which, given an undirected graph $G = (V, E)$ on $n$ vertices and $m  \ge \alpha n^2$ edges and $\eps \in (0,1)$, outputs $x : V \to \{0,1\}$ which has $\val{G}{x} \ge (1-\eps) \cdot \maxval{G}$ in $O(\frac{n}{\eps^2\alpha^2})$ bits of space with probability at least $2/3$.
\end{theorem}
\begin{proof}
We use the $F_0$ estimation sketch of \cite{B2018optimal} which estimates $F_0$ up to a multiplicative $1 \pm \eps'$ error with failure probability $\delta$, using $O\left(\frac{\log(1/\delta)}{\eps'^2} + \log N \right)$ bits of space. Here, the universe is the set of all $N = \binom{n}{2}$ possible edges of the input graph $G$. 

During the stream, we maintain an $F_0$ sketch with $\eps_m = \eps/10$ and $\delta_m = 1/9$ to estimate $m$ which is the number of distinct edges (recall that edges can appear multiple times - we aim to allow repeated edges but ignore their multiplicities). Let $\hat{m}$ denote its estimate for $m$. We feed the stream of inserted edges into a second $F_0$ sketch, instantiating it with $\eps' = \frac{\eps\alpha}{10}$ and failure probability $\delta = \frac{1}{9 \cdot 2^{n}}$. The space used is therefore $O\left(\frac{n}{\eps^2\alpha^2} + \log n\right) = O\left(\frac{n}{\eps^2\alpha^2}\right)$ bits.

For any cut $x : V \to \{0,1\}$, let $\cut(x) \subseteq \binom{V}2$ denote the set of edges crossing the cut in the complete graph on $V$. The number of edges of $G$ crossing the cut is exactly $|E \cap \cut(x)|$. By the inclusion-exclusion principle:
\[
|E \cap \cut(x)| = m + |\cut(x)| - |E \cup \cut(x)|.
\]
After the stream, we iterate over all $2^{n-1}$ possible cuts $x$. For each cut, we instantiate a copy of the $F_0$ sketch's memory state and feed the offline edges of $\cut(x)$ into it. The sketch then outputs an estimate $\hat{\bU}_x$ for $|E \cup \cut(x)|$. We estimate the cut value as $\hat{V}_x = \hat{m} + |\cut(x)| - \hat{\bU}_x$.

By a union bound, with probability at least $2/3$, for all $2^{n-1}$ estimates (one per cut $x$) we have 
that the estimate
\(
\hat \bV_x \coloneqq \hat{m} + |\cut(x)| - \hat \bU_x
\)
approximates the true cut value with additive error
\[
\abs{\hat \bV_x - |E \cap \cut(x)|}
= |\hat{m} - m| + \abs{\hat \bU_x - |E \cup \cut(x)|}
\le \eps_m m + \eps' |E \cup \cut(x)|.
\]
Since $|E \cup \cut(x)| \le \binom{n}{2} \le n^2/2$ for every cut $x$, this yields a uniform additive error bound of at most $\eps_m m + \eps' n^2/2$ per cut.

With $\eps_m = \eps/10$ and $\eps' = \eps \alpha / 10$, using assumption $m \geq \alpha n^2$, we have 
\[
\eps_m m + \frac{\eps' n^2}{2} = \frac{\eps m}{10} + \frac{\eps \alpha n^2}{20} \leq \frac{\eps m}{10} +  \frac{\eps m}{20} \leq \eps m. 
\]
So, every normalized cut value is approximated within additive error at most $\eps$, and selecting the cut $x$ that maximizes $\hat V_x$ yields a $(1-\eps)$-approximation.
\end{proof}

\section{Randomized lower bound for maximum cut in dense graphs (\Cref{thm:maxcut:dense hard})} \label{sec:randomized-dense-LB-max-cut}
In this section, we prove \Cref{thm:maxcut:dense hard}, establishing a linear lower bound on the space of any streaming algorithm
which outputs a $(1-\eps)$-approximate cut given a dense bipartite graph.
\begin{definition}\label{def:maxcut:dense hard:disjoint}
    Let $\gamma \in [0,1]$.
    Two graphs $G,G'$ on $[n]$ are \emph{$(1-\gamma)$-disjoint} if $\Good{\gamma}{G} \cap \Good{\gamma}{G'} = \emptyset$.
\end{definition}

We prove the following combinatorial theorem on the existence of many disjoint instances of $\MaxCut$:

\begin{theorem}\label{thm:maxcut:dense hard:instances}
    There exists $\alpha, C > 0$ such that the following holds.
    For $\epsilon_0 \coloneqq \tfrac14$ and every $0 < \epsilon < \epsilon_0$, let $n \in \N$ be sufficiently large.
    Then there exists a set $\calG$ such that:
    \begin{enumerate}
        \item Each $G \in \calG$ is an undirected simple graph on $[n]$ with $\ge \alpha \binom{n}2$ edges.
        \item $|\calG| \ge 2^{Cn}$.
        \item Each $G \in \calG$ is \emph{bipartite}, i.e., $\maxval{G} = 1$.
        \item For every $G \ne G' \in \calG$, $G$ and $G'$ are $(1-\eps)$-disjoint (in the sense of \Cref{def:maxcut:dense hard:disjoint}).
    \end{enumerate}
\end{theorem}

\Cref{thm:maxcut:dense hard:instances} immediately implies \Cref{thm:maxcut:dense hard}:

\begin{proof}[Proof of \Cref{thm:maxcut:dense hard} assuming \Cref{thm:maxcut:dense hard:instances}]
    We reduce from the $\IdProb$ problem on $\calG$.
    The protocol is simple:
    \begin{itemize}
        \item Alice receives a graph $G^* \in \calG$.
        She runs the hypothesized streaming algorithm for $\MaxCut$ on $G^*$ and sends its state $\bmu$ to Bob.
        \item Bob takes $\bmu$ and runs the $\MaxCut$ algorithm to produce a cut $\bx$.
        Now, say $G \in \calG$ is \emph{good} if $\bx \in \Good{\eps}{G}$.
        If there is a unique good $G \in \calG$, Bob outputs it, and otherwise, he fails.
    \end{itemize}

    We claim that w.p. at least $\delta$ over the randomness of the streaming algorithm,
    Bob successfully identifies Alice's input $G$.
    Indeed, by the algorithm's correctness condition, w.p. at least $\delta$,
    the cut $\bx$ satisfies $\bx \in \Good{\eps}{G^*}$, and therefore $G^*$ is good.

    Now, condition on any $\bx$ such that $G^*$ is good.
    We claim that no other $G \ne G^* \in \calG$ can be good.
    Indeed, if $G \ne G^* \in \calG$ is good, then $\bx \in \Good{\eps}{G}$,
    and in particular $\Good{\eps}{G^*} \cap \Good{\eps}{G} \ne \emptyset$.
    This contradicts the $(1-\eps)$-disjointness promised by \Cref{thm:maxcut:dense hard:instances}.
\end{proof} 

Next, we turn to the proof of \Cref{thm:maxcut:dense hard:instances}.
We reduce this theorem, in turn, to the problem of constructing a large set $\calS \subseteq \{0,1\}^n$ with certain nice properties.
For $x \in \{0,1\}^n$ and $\delta \in [0,1]$, let \[ \HBall{\delta}{x} \coloneqq \{y \in \{0,1\}^n : d(x,y) \le \delta\} \] denote the set of strings of normalized Hamming distance at most $\delta$ from $x$, and let \[ \HPMBall{\delta}{x} \coloneqq \HBall{\delta}{x} \cup \HBall{\delta}{1-x} \] denote the union of these balls around $x$ and its entrywise complement.
(Note also that $B^\delta(1-x) = \{y \in \{0,1\}^n : d(x,y) \ge 1-\delta\}$.)
We prove the following:
\begin{lemma}\label{lemma:maxcut:dense hard:counting}
    For every $0 < \delta < 1/2$, there exists $C > 0$ such that for sufficiently large $n \in \N$, there exists a set $\calS \subseteq \{0,1\}^n$ such that:
    \begin{enumerate}
        \item $|\calS| \ge 2^{Cn}$.
        \item For every $x \in \calS$, $w(x) \in [1/3,2/3]$.
        \item For every $x \ne x' \in \calS$, $\HPMBall{\delta}{x} \cap \HPMBall{\delta}{x'} = \emptyset$.
    \end{enumerate}
\end{lemma}
($w(x)$ denotes the relative Hamming weight of $x$.)
Such a construction does not follow from a greedy packing/covering argument (since the sets $\HPMBall{\cdot}{\cdot}$ are not balls in a metric, unlike the sets $\HBall{\cdot}{\cdot}$), but we can still apply the probabilistic method.

The remainder of this section is split into two subsections: in the first, we reduce proving \Cref{thm:maxcut:dense hard:instances} to proving \Cref{lemma:maxcut:dense hard:counting}, and in the second, we prove \Cref{lemma:maxcut:dense hard:counting}.

\subsection{Reducing to a combinatorial problem}

We first prove \Cref{thm:maxcut:dense hard} assuming \cref{lemma:maxcut:dense hard:counting}.

Let $n \in \N$.
For $S, T \subseteq [n]$ with $S \cap T = \emptyset$, let \[ S \boxtimes T \coloneqq \braces*{ \{u,v\} : u \in S, v \in T} \] denote the ``unordered Cartesian product'' of $S$ and $T$.
Note that this product operation is symmetric, i.e., $S \boxtimes T = T \boxtimes S$.
Also, $|S \uprod T| = |S||T|$.

For $S \subseteq [n]$, let \[ \bip{S} := (n,S \uprod \compS) \] denote the \emph{complete bipartite graph} on $[n]$ between $S$ and $\compS$.
Thus, $\bip{S} = \bip{\compS}$, and $\bip{S}$ has $m = |S||\compS|$ edges.
We will use graphs of this form as our instances in \Cref{thm:maxcut:dense hard:instances}, and we turn to analyzing their $\MaxCut$ structure.

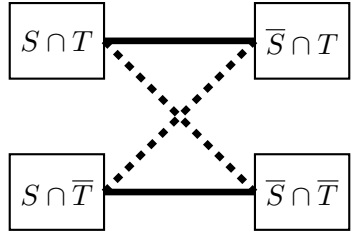
\begin{figure}
\centering
\begin{tikzpicture}
\draw[thick] (-2.25, 2) rectangle (-1, 3) node[midway] {$S \cap T$};
\draw[thick] (-2.25, 0) rectangle (-1, 1) node[midway] {$S \cap \compT$};
\draw[thick] (1, 2) rectangle (2.25, 3) node[midway] {$\compS \cap T$};
\draw[thick] (1, 0) rectangle (2.25, 1) node[midway] {$\compS \cap \compT$};

\draw[thick, line width=2.5pt] (-1, 2.5) -- (1, 2.5);
\draw[thick, line width=2.5pt] (-1, 0.5) -- (1, 0.5);
\draw[thick, dashed, line width=2.5pt] (-1, 2.5) -- (1, 0.5);
\draw[thick, dashed, line width=2.5pt] (-1, 0.5) -- (1, 2.5);
\end{tikzpicture}
\caption{The value of the cut $T$ in the complete bipartite graph $\bip{S}$.
There are edges between $S$ and $\compS$; the thick solid edges stay on the same side of $T$, while the thick dashed edges are cut by $T$.
This cut only has large value if either (i) $S$ is ``very close'' to $T$ (so that $S \cap T$ and $\compS \cap \compT$ are large while $S \cap \compT$ and $\compS \cap T$ are small) or (ii) $S$ is ``very close'' to $\compT$ (so that $S \cap T$ and $\compS \cap \compT$ are small while $S \cap \compT$ and $\compS \cap T$ are large).}\label{fig:maxcut:dense hard}
\end{figure}

Now let $T \subseteq [n]$ be a set, (possibly) distinct from $S$.
We partition the vertices in the graph $\bip{S}$ into four distinct subsets, \[ [n] = \parens*{S \cap T} \sqcup \parens*{S \cap \compT} \sqcup \parens*{\compS \cap T} \sqcup \parens*{\compS \cap \compT}. \]
We similarly partition the edges,
\[ S \uprod \compS = \parens*{(S \cap T) \uprod (\compS \cap T)} \sqcup \parens*{(S \cap T) \uprod (\compS \cap \compT)} \sqcup \parens*{(S \cap \compT) \uprod (\compS \cap T)} \sqcup \parens*{(S \cap \compT) \uprod (\compS \cap \compT)}. \]
Only the second and third types of edges are cut by the cut $T$; thus, the value of the cut $T$ is
\begin{equation}\label{eq:maxcut:dense hard:val}
\val{\bip{S}}{1_T} = \frac{|S| |\compS| - |S \cap T||\compS \cap T| - |S \cap \compT| |\compS \cap \compT|}{|S| |\compS|} = 1 - ((1-p(S,T))q(S,T) + p(S,T) (1-q(S,T))),
\end{equation}
where we define
\begin{align}\label{eq:maxcut:dense hard:p-q}
p(S,T) \coloneqq \frac{|S \cap \compT|}{|S|} &&
\text{ and } &&
q(S,T) \coloneqq \frac{|\compS \cap T|}{|\compS|}.
\end{align}
Qualitatively, the cut value can only be close to $1$ iff either (i) $p$ and $q$ are both very close to $0$, i.e., $S$ and $T$ are almost the same, or (ii) $p$ and $q$ are both very close to $1$, i.e., $S$ and $\compT$ are almost the same.
(See \Cref{fig:maxcut:dense hard} for a visual depiction.)
We now make this qualitative statement quantitative.

\begin{claim}\label{claim:maxcut:dense hard:eps}
    Let $0 < \eps < 1/4$.
    Then $\frac{\eps}{1-\sqrt{\eps}} \le 2 \eps$.
\end{claim}

\begin{proof}
    Cross-multiplies to $\eps \le 2\eps (1-\sqrt{\eps}) \iff 2\eps \sqrt{\eps} \le \eps \iff 2\sqrt{\eps} \le 1 \iff \eps \le \tfrac14$.
\end{proof}

\begin{claim}\label{claim:maxcut:dense hard:pq}
    Suppose $S,T \subseteq [n]$, $p = p(S,T)$ and $q = q(S,T)$ are as in \Cref{eq:maxcut:dense hard:p-q}, and $0 < \eps < 1/4$.
    If $\val{\bip{S}}{1_T} \ge 1-\eps$, then either $\max\{p,q\} \le 2\eps$ or $\min\{p,q\} \ge 1-2\eps$.
\end{claim}
\begin{proof}
    By \Cref{eq:maxcut:dense hard:val}, we have $(1-p)q + p(1-q) \le \eps$.
    Thus, in particular, $(1-p)q \le \eps$ and $p(1-q) \le \eps$.
    Taking the second inequality, we deduce that either $p \le 2\eps$ or $1-q \le 2\eps$.
    Suppose first that $p \le 2\eps$.
    Then the first inequality gives \[ q \le \frac{\eps}{1-\sqrt \eps}, \] which combined with \Cref{claim:maxcut:dense hard:eps} gives $q \le 2\eps$.
    Otherwise, we have symmetrically that $1-p \le 2\eps$.
\end{proof}

Note that the relative Hamming distance between the binary strings $1_S$ and $1_T$ is:
\begin{equation}\label{eq:maxcut:dense hard:hamming}
    d(1_S,1_T) \coloneqq \frac{|S \cap \compT| + |\compS \cap T|}{n} = \frac{|S| \cdot p(S,T) + |\compS| \cdot q(S,T)}{n}.
\end{equation}

We now get:

\begin{lemma}\label{lemma:maxcut:dense hard:hamming}
    Suppose $S,T \subseteq [n]$ and $0 < \eps < 1/4$.
    If $\val{\bip{S}}{1_T} \ge 1-\eps$, then either $d(1_S,1_T) \le 2\eps$ or $d(1_S,1_T) \ge 1-2\eps$.
\end{lemma}

\begin{proof}
    Follows immediately from \Cref{eq:maxcut:dense hard:hamming,claim:maxcut:dense hard:pq} (and $|S| + |\compS| = n$).
\end{proof}

Finally, we use \Cref{lemma:maxcut:dense hard:hamming} to prove \Cref{thm:maxcut:dense hard:instances} assuming \Cref{lemma:maxcut:dense hard:counting}.

\begin{proof}[Proof of \Cref{thm:maxcut:dense hard:instances} assuming \Cref{lemma:maxcut:dense hard:counting}]
    Let $\calS \subseteq \{0,1\}^n$ be the set of strings promised by \Cref{lemma:maxcut:dense hard:counting} for $\delta = 2\eps$.
    For each string $x \in \{0,1\}^n$, we let $S = \supp(x)$ denote the support of $x$, and we create a corresponding graph $G = \bip{S}$.
    We let $\calG$ denote the set of graphs so created, and we claim that $\calG$ fulfills the desiderata of \Cref{thm:maxcut:dense hard:instances}.
    Indeed:
    \begin{enumerate}
        \item For each $G \in \calG$, the number of edges in $G = \bip{S}$ for $S = \supp(x)$ is $|S|(n-|S|)$, and by assumption on $w(x)$, $\frac13n \le |S| \le \frac23n$.
        \item $|\calG| = |\calS| \ge 2^{Cn}$.
        \item Each $G \in \calG$ is bipartite by definition.
        \item We show that for every $y \in \{0,1\}^n$, there do not exist $G \ne G' \in \calG$ such $\val{G}{y} \ge 1-\eps$ and $\val{G'}{y} \ge 1-\eps$.
        Let $S \coloneqq \supp(x)$, $S' \coloneqq \supp(x')$, and $T \coloneqq \supp(y)$.
        Let $G \coloneqq \bip{S}$ and $G' \coloneqq \bip{S'}$.
        Then $1_T = y$ and $\val{\bip{S}}{1_T} \ge 1-\eps$ and $\val{\bip{S'}}{1_T} \ge 1-\eps$.
        We apply \Cref{lemma:maxcut:dense hard:hamming} to deduce that (i) either $d(x,y) \le 2\eps$ or $d(x,y) \ge 1-2\eps$ and (ii) either $d(x',y) \le 2\eps$ or $d(x',y) \ge 1-2\eps$.
        Thus, $y \in \HPMBall{2\eps}{x}$ and $y \in \HPMBall{2\eps}{x'}$, contradicting the third guarantee of \Cref{lemma:maxcut:dense hard:counting}.
    \end{enumerate}
\end{proof}

\subsection{Solution to the combinatorial problem}

We recall two useful simple properties of the binary entropy function: (i) $\entropy(\delta)$ is strictly increasing over the interval $[0,1/2]$ and (ii) $\entropy(1/2) = 1$ (so $\entropy(\delta) < 1$ for $\delta \in [0,1/2)$).

\begin{proof}[Proof of \Cref{lemma:maxcut:dense hard:counting}]
    Let $T$ be a size parameter to be chosen later. We claim that a set $\calS \subseteq \{0,1\}^n$ of $T$ independent and uniformly random strings fulfills the desiderata, even if $T$ is exponentially large.

    Firstly, we observe that for every $x \in \{0,1\}^n$, $|\HBall{\delta}{x}| \le K \coloneqq 2^{n\entropy(\delta)}$, where $\entropy(\delta) = -p\log_2 p - (1-p) \log_2 (1-p)$ is the binary entropy function.
    This follows from a simple counting argument: $|\HBall{\delta}{x}| = |\HBall{\delta}{0}|$ by translation symmetry; $|\HBall{\delta}{0}|$ is the number of $n$-binary strings of fractional Hamming weight $\le \delta$.
    Thus, we have by \Cref{lemma:prelim:hamming sum} (and increasingness of $\entropy(\cdot)$) that
    \[ |\HBall{\delta}{0}| = \sum_{i=0}^{\lfloor \delta n \rfloor} \binom{n}{i} \le 2^{n\entropy(\lfloor \delta n \rfloor /n)} \le 2^{n\entropy(\delta)} = K \] since $\delta < 1/2$.

    Now we have immediately that $|\HPMBall{\delta}{x}| \le |\HBall{\delta}{x}| + |\HBall{\delta}{1-x}| = 2K$ (indeed, this will be an equality since $\delta < 1/2$).

    Next, we consider a fixed $y \in \{0,1\}^n$ and uniformly random $\bx \in \{0,1\}^n$: By the previous paragraph and symmetry, $\Pr_\bx[y \in \HPMBall{\delta}{\bx}] = \Pr[\bx \in \HPMBall{\delta}{y}] = \frac{2K}{2^n}$.

    Now consider fixed $y \in \{0,1\}^n$ and $T$ i.i.d. uniform strings $\bx_1,\ldots,\bx_T \in \{0,1\}^n$.
    By independence and a union bound, \[ \Pr_{\bx_1,\ldots,\bx_T} [\exists i \ne j \in [T] : y \in \HPMBall{\delta}{\bx_i} \cap \HPMBall{\delta}{\bx_j}] \le \parens*{\frac{2TK}{2^n}}^2. \]

    Finally, consider taking a union bound over all $y \in \{0,1\}^n$:
    We deduce \[ \Pr_{\bx_1,\ldots,\bx_T} [\exists y \in \{0,1\}^n, i \ne j \in [T] : y \in \HPMBall{\delta}{\bx_i} \cap \HPMBall{\delta}{\bx_j}] \le \frac{(2TK)^2}{2^n}. \]
    This probability is strictly less than $1$ if $T^2 < \frac{2^n}{(2K)^2} = 2^{n - 2(1+n\entropy(\delta))} = 2^{-2 + (1-2\entropy(\delta))n}$,
    i.e., if $T <  2^{-1 + (1-2\entropy(\delta))n/2}$.
    As long as $\delta$ is sufficiently small, $(1-2\entropy(\delta))/2 > 0$ (i.e., $\entropy(\delta) < 1/2$) and so $T$ can therefore be exponentially large.

    (We can also take a union bound over $\bx_1,\ldots,\bx_T$ to guarantee that $w(\bx_i) \in [1/3,2/3]$ for each $i \in [T]$. 
    This happens except w.p. $2^{\Omega(-n)}$ where the constant does not depend on $\delta$, and so for sufficiently small $\delta$ we can do a union bound.)
\end{proof}

\section{Deterministic lower bound for maximum cut value (\Cref{thm:deterministic-max-cut})} \label{sec:deter-LB-max-cut}

In this section we prove \Cref{thm:deterministic-max-cut}, establishing an
$\Omega(n \log n / \varepsilon^{2})$ space lower bound for deterministic streaming
algorithms that output a $(1\pm\varepsilon)$-approximation to the \emph{value} of
the maximum cut.


Our proof follows the same high-level strategy as the randomized lower bound
from \Cref{sec:randomized-LB-max-cut}, but differs in two important ways:
(i) we work with deterministic algorithms, and
(ii) we assume access to \emph{value} queries instead of cut queries.

\subsection{Constructing the hard family}

\begin{lemma}\label{lemma:det:hard_family}
    For every $k \in \N$ and $\eta > 0$, there exists $\Csize > 0$ such that for sufficiently large $n$,
    there is a family $\calF$ of bipartite $k$-regular graphs on the same vertex-set $L \sqcup R$ with $|L| = |R| = k$ such that
    \begin{enumerate}
        \item \emph{Bipartite regularity:} Every $G \in \calF$ is a bipartite $k$-regular graph $L \sqcup R$.
        \item \emph{Mild differences between graphs:} 
        For every $G_1 \ne G_2 \in \calF$, $|E(G_1) \triangle E(G_2)| \ge \eta nk$.
        \item \emph{Size: $\log_2 |\calF| \ge \Csize nk \log n$.}
    \end{enumerate}
\end{lemma}

This can be proved similarly to \Cref{lem:hard_family}, by considering the distribution
$\mathcal{G}(n,k)$ corresponding to the union
of $k$ independent uniformly random perfect bipartite matchings on $L \sqcup R$,
where $L$ and $R$ are two fixed (disjoint) vertex sets of size $n$.

\subsection{The augmentation gadget for maximum cuts}

Let $A$ be a graph on vertex set $V$ with $|V|=N$, and let $x \in \{\pm1\}^{V}$.
We define an augmentation gadget that encodes the cut value $\val{A}{x}$ into the
maximum cut value of a larger graph.
Specifically, for any $x \in \{\pm1\}^{V}$, we create a new graph on vertex-set $V \sqcup \{s_{+1},s_{-1}\}$,
where $s_{+1},s_{-1}$ are two new ``sink'' vertices.
Specifically, $H_x$ has an edge between $s_{+1}$ and $s_{-1}$ with weight $5000Nk$,
and then, for every $w \in L \sqcup R$, an edge between $s_{-x_w}$ and $w$  with weight $100k$.

\begin{lemma}
\label{lemma:det:gadget-correctness}
For every $k$-regular graph $A$ on vertex-set $V$ with $|V|=N$ and cut assignment $x \in \{\pm1\}^V$, we have:
\[
\opt{A \cup H_x} = 5000Nk + 100Nk +
\val{A}{x}.
\]
\end{lemma}

\begin{proof}
    We first show that $\opt{A \cup H_x} \ge 5000Nk + 100 k \cdot  N + \val{A}{x}$.
    Specifically, consider the cut $x^* \in \{\pm1\}^{V \sqcup \{s_{\pm1}\}}$ which extends $x$ via $x'(s_{+1}) = +1, x'(s_{-1}) = -1$.
    By construction, $x^*$ cuts the sink-sink edge (with weight $5000Nk$). 
    Furthermore, it cuts the edges between $s_{-x_w}$ and $w$ for every $w \in V$,
    since $x^*_{s_{-x_w}} = -x_w$ and $x^*_w = x_w$.
    This adds to the value $100k \cdot N$. 
    Finally, the only other edges cut are those which $x$ also cuts in $A$, which is exactly $\val{A}{x}$. 

    We now show that $x^*$ is actually an optimum cut in $A \cup H_x$.
    Suppose this is not true, and that $\val{A \cup H_x}{x'} > \val{A \cup H_x}{x^*}$ for some cut $x' \in \{\pm1\}^{V \sqcup \{s_{\pm1}\}}$.
    We proceed in cases to get a contradiction. 
    \begin{itemize}
        \item First suppose that $x'_{s_{+1}} = x'_{s_{-1}}$,
        so that the sink-sink edge is not cut.
        Consider flipping the assignment of $s_{+1}$.
        By construction, the sink-sink edge is now cut, adding $5000Nk$ to the cut value. 
        At the same time, it can decrease the cut value by at most $100k \cdot N$,
        since this is the maximum total weight of edges connected to $s_{+1}$. 
        Since we have $5000Nk - 100Nk > 0$, it must be that $x'_{s_{+1}} \ne x'_{s_{-1}}$. 
        WLOG, we have $x'_{s_{+1}} = +1$ and $x'_{s_{-1}} = -1$. 
        \item Now suppose that for some $w \in V$, 
        we have $x'_w \ne x_w$ (so that the edge between $s_{-x_w}$ and $w$ is not cut). 
        If we flip the assignment of $w$,
        we increase the cut value by $100k$ (the weight of the edge between $s_{-x_w}$ and $w$),
        and can decrease it by at most $k$ (the weight of all other edges incident to $w$, since $A$ is a $k$-regular graph). 
        Hence, we always have $x'_w = x_w$, as desired.
        \qedhere
    \end{itemize}
\end{proof}

\subsection{Reducing graph distinguishing to value separation}

\begin{lemma}\label{lemma:det:value-gap}
    Let $A_1,A_2$ be two $k$-regular bipartite matchings on $L \sqcup R$.
    If $A_1, A_2$ differ in $\eta nk$ edges, then there exists a cut $z \in \{\pm1\}^{L \sqcup R}$ such that $|\val{A_1}{z} - \val{A_2}{z}| > \Omega(\eta n \sqrt{k})$.
\end{lemma}
\begin{proof}
Let $D\coloneqq E(A_1)\triangle E(A_2)$ and $m\coloneqq|D|$.
Define an $n\times n$ signed matrix $B$ (rows indexed by $L$, columns by $R$) by
\[
B_{uv} =
\begin{cases}
+1 & \text{if } (u,v)\in E(A_1)\setminus E(A_2),\\
-1 & \text{if } (u,v)\in E(A_2)\setminus E(A_1),\\
0  & \text{otherwise.}
\end{cases}
\]

View a cut $z \in \{\pm1\}^{L \sqcup R}$ as a pair
$(x\in\{\pm1\}^L$, $y\in\{\pm1\}^R)$.
Using the identity $\mathbf{1}[s_ut_v=-1]=(1-s_ut_v)/2$ and the fact that
$|E(A_1)|=|E(A_2)|=nk$, we have
\[
\val{A_1}{z}-\val{A_2}{z} = -\frac12 x^\intercal B y.
\]
Hence it suffices to find $x,y$ such that $|x^\top B y|=\Omega(n\sqrt{k})$.

Pick $\by\in\{\pm1\}^R$ uniformly at random.
For each $u\in L$, let
\[
\bZ_u \coloneqq (B\by)_u = \sum_{v\in R} B_{uv}\, \by_v .
\]
Let $d_u \coloneqq \sum_v B_{uv}^2$, the number of nonzeros in row $u$.
Then $\bZ_u$ is a Rademacher sum with coefficient $\ell_2$-norm $\sqrt{d_u}$.
By the Khintchine inequality for $p=1$,
\[
\mathbb{E}_t[|\bZ_u|] \ge c\,\sqrt{d_u}
\]
for an absolute constant $c>0$.
Summing over all $u\in L$ and using linearity of expectation,
\[
\mathbb{E}_\bt\|B\bt\|_1
= \sum_u \mathbb{E}_\bt|\bZ_u|
\ge c \sum_u \sqrt{d_u}.
\]
Therefore there exists $y^\star\in\{\pm1\}^R$ such that
\[
\|Bt^\star\|_1 \ge c \sum_u \sqrt{d_u}.
\]

Since each left vertex has degree $k$ in each of $A_1$ and $A_2$, we have
$d_u\le 2k$ for all $u$.
Using $d_u\le \sqrt{2k}\sqrt{d_u}$ and summing over $u$ gives $m=\sum_u d_u \le \sqrt{2k}\sum_u \sqrt{d_u}$,
and hence
$\sum_u \sqrt{d_u} \ge \frac{m}{\sqrt{2k}}$.
Combining with the previous bound,
\[
\|By^\star\|_1 \ge \frac{c}{\sqrt{2}}\,\frac{m}{\sqrt{k}}.
\]

Now define $x^\star\in\{\pm1\}^L$ coordinatewise by
$x_u^\star\coloneqq \mathrm{sign}((By^\star)_u)$.
Then
\[
x^{\star\top} B y^\star
= \sum_u x_u^\star (By^\star)_u
= \sum_u |(By^\star)_u|
= \|By^\star\|_1.
\]
Thus
\[
|x^{\star\top} B y^\star|
\ge \frac{c}{\sqrt{2}}\,\frac{m}{\sqrt{k}}. \qedhere
\]
\end{proof}

\subsection{Completing the lower bound}

The lower bound follows from \Cref{lemma:det:value-gap,lemma:det:hard_family,lemma:det:gadget-correctness} and setting $k \coloneqq 1/\epsilon^2$.
Specifically, every candidate pair of graphs $G_1 \ne G_2 \in \calF$ differs in $\Omega(n/\epsilon^2)$ edges,
and therefore by \Cref{lemma:det:value-gap}, there is some cut $z$ on which $G_1$ and $G_2$ have maximum cuts differing in $\Omega(n/\epsilon)$ by weight.
This discrepancy can be detected via a single application of an $\epsilon$-approximate maximum cut oracle.

\section*{Acknowledgements}
Hoai-An Nguyen is supported in part by NSF GRFP grant number DGE2140739, Office of Naval Research award number N000142112647, and a Simons Investigator Award. David P. Woodruff is supported in part by Office of Naval Research award number N000142112647 and a Simons Investigator Award. We thank the anonymous ICALP reviewers for their valuable feedback.

\printbibliography
\appendix
\section{Extended preliminaries}
\subsection{Expander graphs}

An \emph{$(N, d, \lambda)$-graph} is a $d$-regular graph on $N$ vertices in which the absolute value of every nontrivial eigenvalue in the normalized Laplacian is at most $\lambda$.

\begin{theorem}[Strongly explicit expander construction, {\cite[Theorem~1.2]{Alo21}}] \label{thm:seExpander}
    For any prime $p \equiv 1 \mod 4$ and every sufficiently large $N$ there is a strongly explicit construction of an $(N, d, \lambda)$-graph where $d = p+2$ and $\lambda \le (1 + \sqrt{2})\sqrt{d-1}/d + o(1)$, and the $o(1)$ term tends to zero as $N$ tends to infinity.
(Strongly explicit means that the adjacency list of any vertex can be produced in time $\poly \log N$.)
\end{theorem}

\begin{theorem}[Bertrand's postulate for primes equivalent to $1$ mod $4$, \cite{Bre32}] \label{thm:primes}
    For every $n \ge 7$, there exists a prime $p \in [n,2n]$ such that $p \equiv 1 \pmod 4$.
\end{theorem}

\begin{theorem}[Chernoff bound for expander walks, {\cite[Theorem~4.22]{Vad12}}]
\label{thm:cherExp}
Let $G = (V,E)$ be an $(n,d,\lambda)$-graph and let $f : V \to [0,1]$ be any function.
Consider a random walk $\bw_1,\ldots,\bw_t$ in $G$ from a random start vertex $\bw_1 \sim \Unif{V}$.
Then for every $\eps > 0$,
\[
\Pr \bracks*{\abs*{ \frac1t \sum_{i=1}^t f(\bw_i) - \mu(f) } \ge \lambda(G) + \eps} \le 2\exp(-\Omega(\eps^2 t)),
\]
where $\mu(f) \coloneqq \Exp_{\bv \sim \Unif{V}}[f(\bv)]$ is the average value of $f$ on a uniformly random vertex in $V$.
\end{theorem}

\subsection{$\PermProb$ communication problem}

\begin{definition}\label{def:perm prob}
    Let $n \in \N$.
    In the $\PermProb$ problem, Alice receives a permutation $\pi : \{0,1\}^n \to \{0,1\}^n$ and Bob receives a pair $(x,i)$ for $x \in \{0,1\}^n$ and $i \in [n]$ and must output $(\pi(x))_i$.
\end{definition}

Note that $\PermProb$ is precisely the standard $\IndProb$ problem applied to the concatenated string $\pi(0^n) \pi(0^{n-1}1) \cdots \pi(1^n)$ with the additional promise that the input string represents a permutation.

\begin{theorem}[{\cite[Lemma 1]{SW15}}]\label{thm:perm prob}
     In any one-way protocol for $\PermProb$ with public randomness and constant advantage,
     Alice must send Bob $\Omega(2^n \cdot n)$ bits of information.
\end{theorem}

\section{Alternate sampling algorithm for dense streams} \label{appen:altSamp}
We present our algorithm \hyperref[alg:buildA]{subsampling-dense-streams} (\Cref{alg:buildA}) which we will use to provide alternate algorithms for all CSPs (including $\MaxCut$), $\DenSub$, $\Sim$, and $\Rare{k}$. We emphasize that the update time is only (amortized) $\polylog n$. However, it only works for simple graphs. 

\begin{theorem}[Sampling algorithm for dense streams] \label{thm:sampAlg}
For every $t, N \in \N$ and $\lambda > 0$, there exists a randomized insertion-only streaming algorithm which, given a stream of elements $\calM \subseteq[N]$ with $|\calM| \ge \alpha N$, outputs a sketch $\bsigma$ in $O\left(t \log(1/\lambda) + \log N\right)$ bits of space and $O(\log N)$ amortized update time with the following property. For every function $f : [N] \to [0,1]$ and $\eps > 0$, $\bsigma$ can be used to produce an estimate $\hat{\bmu}(f) \in [0,1]$ for 
$\mu_\calM(f) \coloneqq \Exp_{w \sim \Unif\calM} f(w)$ s.t.
\[
\Pr_{\bsigma}\left[|\hat{\bmu}(f) - \mu_{\mathcal{M}}(f)| \ge \eps+ \frac{\lambda}{\alpha}\right] \le 2\exp(-\Omega(\eps^2 \alpha^2 t)).
\]
($\hat{\bmu}(f)$ is a \emph{deterministic} function of $\bsigma$.)
\end{theorem}

\begin{algorithm}[t]
\caption{subsampling-dense-streams ($N, t, \lambda$)}
\label{alg:buildA}
\begin{algorithmic}[1]
    \State \textbf{Before the stream:}
    \State Take $d$ to be the smallest integer that equals $p+2$ for a prime $p \equiv 1 \mod 4$ and meets condition $(1 + \sqrt{2}) \sqrt{d-1} / d \ge \lambda $. 
    \State Fix a regular expander $\expand$ with $N$ vertices and degree $d$ using \Cref{thm:seExpander}. 
    \State Sample a uniformly random vertex $\bw_1$ in $\expand$ and then perform a random walk from $\bw_1$ in $\expand$ to get $\bw_2,\ldots,\bw_t$. Write $\bW = (\bw_1,\dots,\bw_t)$.
    \State Store an associated counter $\bc_j$ for each $j \in [t]$.
    \State Store a counter $c_s$ to compute $|\calM|$. 
    \State \textbf{During the stream:}
    \State Increment $c_s$ if the update is an insertion and decrement $c_s$ if the update is a deletion. 
    \For{every $z = \frac{t}{\log N}$ chunk of updates}
        \State Store the $z$ updates.
        \State Hash the insertions to $2z$ buckets using a pairwise-independent hash function.
        Denote this structure as $\bZ_1$.
        \State Hash the deletions to $2z$ buckets  using a pairwise-independent hash function.
        Denote this structure as $\bZ_2$. 
        \For{each $j \in [t]$}
            \State Suppose $C_1$ insertions of $\bw_j$ are in $\bZ_1$.
            Perform $\bc_j \from \bc_j + C_1$. 
            \State Suppose $C_2$ deletions of $\bw_j$ are in $\bZ_2$.
            Perform $\bc_j \from \bc_j - C_2$. 
        \EndFor 
    \EndFor 
    \State \textbf{After the stream:}
    \State Initialize $\bsigma$ to be empty. 
    \State For each $j \in [t]$, look at the counter $c_j$.
    Add $\max\{c_j,0\}$ copies of $\bw_j$ to $\bsigma$.
    \State Return $\bsigma$.
    \State \textbf{To sketch a function:}
    \State Get value of $|\calM|$ from $c_s$. 
    \State For any $f:[N] \to [0,1]$, output the estimate
    $\hat{\bmu}(f) \gets \frac{N}{t|\calM|} \cdot \sum_{i \in \bsigma} f(i)$.
\end{algorithmic}
\end{algorithm}

We analyze the space and update time.
We first fix a regular expander $\expand$ with appropriate degree $d$.
By \Cref{thm:primes}, we have that there is always a suitable prime $p$ such that we have $d = p+2$ and $d = O(1/\lambda^2)$.
We then perform random walk $\bW$ on $\expand$ in the following way.

We pick a vertex of $\expand$ uniformly at random to start at.
Let us call this vertex $w_1$.
We can store $\bw_1$ in $O(\log N)$ space.
Now, we can produce the adjacency list of $\bw_1$ since $\expand$ is strongly explicit.
We pick one of $\bw_1$'s neighbors uniformly at random to be $\bw_2$.
We can store $\bw_2$ in $O(\log(1/\lambda))$ space since we only have to store what $i^{\text{th}}$ neighbor of $\bw_1$ it is.
We continue this process to produce and store a walk of length $t$.
So, the space to store the walk is only $O(t \log(1/\lambda))$.
We store an exact counter $c_e$ for each vertex in $\bW$.
Each counter will have value between $0$ and $O(B)$ where $B$ is the maximum multiplicity of an element.
Since $B$ is a constant, the counters only incur an additional $O(t)$ space. 

We hold $\bW$ and the associated counters throughout the stream.
We also hold a counter to keep value $|\calM|$ which only uses space $O(\log N)$ since we have that the length of the stream is at most $\poly(N)$.
During the stream we do the following for every $z = t/ \log N$ chunk of updates.
We store the $z$ updates.
This uses $O(t)$ space.
Then we hash the insertions and deletions to $\bZ_1$ and $\bZ_2$ using two pairwise independent hash functions.
The hash functions only incur $O(\log N)$ space.
Now, we walk through $\bW$.
For each vertex in $\bW$, we check whether it is in $\bZ_1$ and/or $\bZ_2$ and adjust the associated counter accordingly.

We analyze the update time of one batch of $z$ updates.
Using the guarantee the pairwise independent hash functions, we have that the probability of a collision between two items in $\bZ_1$ (or $\bZ_2$) is $1/(4z^2)$.
Consider one item $x$.
The expected number of collisions of other items with $x$ is at most $1/(4z)$.
The expected number of collisions among the $z$ items is therefore at most $1/4$.
So the expected update time of checking whether each vertex in $\bW$ is in $\bZ_1$ or $\bZ_2$ for this batch of $z$ updates is $O(t)$.
If the total number of updates in the stream is $U$, then the expected update time over the entire stream is at most $O(t) \cdot U/z$.
By Markov's bound, the actual total update time will not exceed $6 \cdot O(t) \cdot U/z$ with probability at least $5/6$.
So with probability at least $5/6$, we have an amortized update time of $O(\log n)$.

We now proceed with the analysis.
Say we want to query the average value of the function $f:[N] \to [0,1]$ over elements in the set $\calM$.
We define the function $g:[N] \to [0,1]$ via
\[
g(w) \coloneqq \begin{cases} f(w) \cdot \frac{b_w}{B} & w \in \calM \\ 0 & \text{otherwise}. \end{cases}
\]
where $B$ is the max multiplicity of an element and $b_w$ is the multiplicity of element $w$ in set $\calM$. 
Note that
\begin{equation}\label{eq:sampler:g-vs-f}
\sum_{w=1}^N g(w) = \sum_{w\in \calM} f(w).
\end{equation}
Note that set $\calM$ includes multiplicity whereas $N$ does not, but this is accounted for in the re-scaling in function $g$.

Using the expander Chernoff bound (\Cref{thm:cherExp}),  we have the following for $g: [N] \to [0,1]$: 
\begin{align*}
\Pr \bracks*{\abs*{ \frac{N}{t|\calM|} \sum_{i=1}^t g(\bw_i) - \frac1{|\calM|} \sum_{w \in \calM} f(w) } \ge \eps + \frac{\lambda}{\alpha}}
&\le \Pr \bracks*{\abs*{ \frac{N}{t|\calM|} \sum_{i=1}^t g(\bw_i) - \frac1{|\calM|} \sum_{w \in \calM} f(w) } \ge \eps + \frac{N}{|\calM|} \lambda} \\
&= \Pr \bracks*{\abs*{ \frac1t \sum_{i=1}^t g(\bw_i) - \frac1N \sum_{w=1}^N g(w) } \ge \eps \alpha + \lambda} \\
&\le 2\exp\left(-\Omega\left(\eps^2\alpha^2 t\right)\right). 
\end{align*}
Note that $\frac1{|\calM|} \sum_{w \in \calM} f(w)$ is the mean of $f$ over $\calM$, while $\frac1N \sum_{w=1}^n g(w)$ is the mean of $g$ over $N$.

\section{Densest subgraph bounds} \label{sec:DS}
We first formally set-up the problem. 
For an undirected $G = (V,E)$ and $S \subseteq V$ (with $S \ne \emptyset$), we define the \emph{density} of $S$ as
\[
\den{G}{S} \coloneqq \frac{|E[S]|}{|S|}
\]
where $G[S]$ is the induced subgraph on $S$ and $|E[S]|$ is the number of edges in $G$ whose endpoints are both in $S$.
($\den{G}{S}$ is the same as twice the average degree of this induced subgraph.)
We define
\[
\maxden{G} \coloneqq \max_{\substack{S \subseteq V \\ S \ne \emptyset}} \den{G}{S}
\]
as the maximum density of any subgraph.
The \emph{$\eps$-approximate $\DenSub$ problem} is to, given $G$, output a set $S \subseteq V$ of vertices such that $\den{G}{S} \ge (1-\eps) \cdot \maxden{G}$ for a given $\eps \in (0,1)$; again, there is also a value approximation version of this problem (output $v$ s.t. $v \in (1\pm\eps) \cdot \maxden{G}$).
\subsection{Randomized lower bound for sparse graphs}

\begin{theorem}\label{thm:densub:sparse hard}
    For constant $\eps>0$, any randomized streaming algorithm which returns $v \in [0,1]$ such that $(1-\eps) \cdot \maxden{G} \le v \le \maxden{G}$ or returns a set $S \subseteq V$ such that $\den{G}{S} \geq (1-\eps) \cdot \maxden{G}$ with probability at least $2/3$ requires $\Omega(n \log n)$ bits of space. 
\end{theorem} 
Here we prove \Cref{thm:densub:sparse hard}. We first give the bound on outputting the \emph{value}. We use the same hard instance and a similar reduction to prove the bound on outputting the \emph{solution}.

\begin{figure}
\centering
\begin{tikzpicture}
\coordinate (00lt) at (0,7);
\coordinate (00rt) at (1,7);
\coordinate (00lb) at (0,6);
\coordinate (00rb) at (1,6);
\coordinate (01lt) at (0,5);
\coordinate (01rt) at (1,5);
\coordinate (01lb) at (0,4);
\coordinate (01rb) at (1,4);
\coordinate (10lt) at (0,3);
\coordinate (10rt) at (1,3);
\coordinate (10lb) at (0,2);
\coordinate (10rb) at (1,2);
\coordinate (11lt) at (0,1);
\coordinate (11rt) at (1,1);
\coordinate (11lb) at (0,0);
\coordinate (11rb) at (1,0);

\fill (00lt) circle (3pt);
\fill (00rt) circle (3pt);
\fill (00lb) circle (3pt);
\fill (00rb) circle (3pt);
\fill (01lt) circle (3pt);
\fill (01rt) circle (3pt);
\fill (01lb) circle (3pt);
\fill (01rb) circle (3pt);
\fill (10lt) circle (3pt);
\fill (10rt) circle (3pt);
\fill (10lb) circle (3pt);
\fill (10rb) circle (3pt);
\fill (11lt) circle (3pt);
\fill (11rt) circle (3pt);
\fill (11lb) circle (3pt);
\fill (11rb) circle (3pt);

\node[left=4pt]  at (00lt) {$(00,\vLeft,\vTop)$};
\node[right=4pt] at (00rt) {$(00,\vRight,\vTop)$};
\node[left=4pt]  at (00lb) {$(00,\vLeft,\vBot)$};
\node[right=4pt] at (00rb) {$(00,\vRight,\vBot)$};
\node[left=4pt]  at (01lt) {$(01,\vLeft,\vTop)$};
\node[right=4pt] at (01rt) {$(01,\vRight,\vTop)$};
\node[left=4pt]  at (01lb) {$(01,\vLeft,\vBot)$};
\node[right=4pt] at (01rb) {$(01,\vRight,\vBot)$};
\node[left=4pt]  at (10lt) {$(10,\vLeft,\vTop)$};
\node[right=4pt] at (10rt) {$(10,\vRight,\vTop)$};
\node[left=4pt]  at (10lb) {$(10,\vLeft,\vBot)$};
\node[right=4pt] at (10rb) {$(10,\vRight,\vBot)$};
\node[left=4pt]  at (11lt) {$(11,\vLeft,\vTop)$};
\node[right=4pt] at (11rt) {$(11,\vRight,\vTop)$};
\node[left=4pt]  at (11lb) {$(11,\vLeft,\vBot)$};
\node[right=4pt] at (11rb) {$(11,\vRight,\vBot)$};

\draw[very thick] (00lt) -- (01rt);
\draw[very thick] (00lb) -- (01rb);
\draw (01lt) -- (00rt);
\draw (01lb) -- (00rb);
\draw (10lt) -- (11rt);
\draw (10lb) -- (11rb);
\draw (11lt) -- (10rt);
\draw (11lb) -- (10rb);

\draw[very thick] (00lt) -- (00lb);

\draw[very thick] (01rt) -- (01rb);
\draw (11rt) -- (11rb);
\end{tikzpicture}
\caption{The output of the reduction in \Cref{alg:densub:sparse hard:alice-reduction,alg:densub:sparse hard:bob-reduction} for $n=2$, $\pi :00 \mapsto 01, 01 \mapsto 00, 10 \mapsto 11, 11 \mapsto 10$, $x = 00$, and $i=2$. Alice contributes the $2 \cdot 2^n = 8$ diagonal edges and Bob contributes the $1 + 2^{n-1} = 3$ vertical edges. In this case, there is a four-cycle (highlighted with thick edges), corresponding to the fact that $\pi(x)_i = (01)_2 = 1$.}\label{fig:densub:sparse hard}
\end{figure}
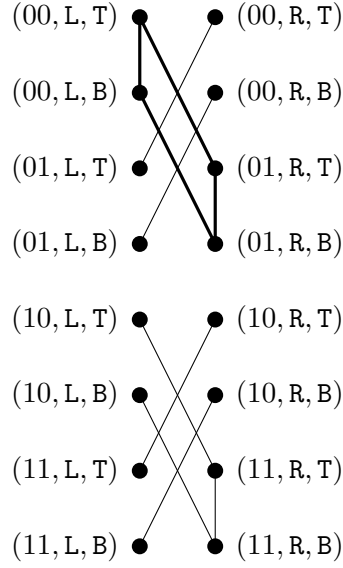

It is useful to consider the following graph-theoretic interpretation of $\PermProb$. We view Alice's input permutation $\pi : \{0,1\}^n \to \{0,1\}^n$ as a \emph{bipartite matching} $M$ whose left and right vertex sets are both $\{0,1\}^n$ and each $(x,\pi(x))$ is an edge.
Bob's input $(x,i)$ then corresponds to a single left vertex $x$ and the set $S_i \coloneqq \{y \in \{0,1\}^n : y_i=1\}$ of right vertices, and his goal is to determine whether Alice's matching $M$ matches $x$ on the left with any vertex in $S_i$ on the right.
Consider the following communication-to-streaming reduction. Fix $n \in \N$. For $a \in \{\vLeft,\vRight\}$, $b \in \{\vTop,\vBot\}$, let $V^{a,b} \coloneqq \{0,1\}^n \times \{(a,b)\}$. (I.e., each element of $V^{a,b}$ is of the form $(x,a,b)$ for $x \in \{0,1\}^n$.) Let $V \coloneqq \bigcup_{a \in \{\vLeft,\vRight\}} \bigcup_{b \in \{\vTop,\vBot\}} V^{a,b}$. Alice and Bob jointly create a graph $G$ on vertex-set $V \coloneqq \bigcup_{a \in \{\vLeft,\vRight\}} \bigcup_{b \in \{\vTop,\vBot\}} V^{a,b}$ by running \Cref{alg:densub:sparse hard:alice-reduction,alg:densub:sparse hard:bob-reduction} below, respectively.
An example output of these algorithms is written in \Cref{fig:densub:sparse hard}.

\begin{algorithm}[t]
\caption{Alice's protocol reducing $\PermProb$ to $\DenSub$}\label{alg:densub:sparse hard:alice-reduction}
    \begin{algorithmic}[1]
        \Require $n \in \N$ and a permutation $\pi : \{0,1\}^n \to \{0,1\}^n$.
        \For{$y \in \{0,1\}^n$}
        \State Create a ``horizontal'' edge $\{(y,\vLeft,\vTop),(\pi(y),\vRight,\vTop)\}$.
        \State Create a ``horizontal'' edge $\{(y,\vLeft,\vBot),(\pi(y),\vRight,\vBot)\}$.
        \EndFor
    \end{algorithmic}
\end{algorithm}

\begin{algorithm}[t]
\caption{Bob's protocol reducing $\PermProb$ to $\DenSub$}\label{alg:densub:sparse hard:bob-reduction}
    \begin{algorithmic}[1]
        \Require $n \in \N$ and input $(x,i)$ for $x \in \{0,1\}^n$ and $i \in [n]$.
        \State Create a ``vertical'' edge $\{(x,\vLeft,\vTop),(x,\vLeft,\vBot)\}$.
        \For{$z \in \{0,1\}^n$ such that $z_i = 1$}
        \State Create a ``vertical'' edge $\{(z,\vRight,\vTop),(z,\vRight,\vBot)\}$.
        \EndFor
    \end{algorithmic}
\end{algorithm}

We claim the following about the output of the reduction:

\begin{claim} \label{cl:denseSquare}
Let $n \in \N$, $\pi : \{0,1\}^n \to \{0,1\}^n$ be a permutation, and $x \in \{0,1\}^n$ and $i \in [n]$. Let $G$ be the graph created if Alice runs \Cref{alg:densub:sparse hard:alice-reduction} with input $\pi$ and Bob runs \Cref{alg:densub:sparse hard:bob-reduction} with input $(x,i)$. Then:
\begin{itemize}
    \item If $(\pi(x))_i = 1$, then $G$ contains exactly one square (i.e., a $4$-cycle), and the rest of the connected components are single edges and length-$3$ paths; hence, $\maxden{G} \ge 1$.
    \item If $(\pi(x))_i = 0$, then the connected components of $G$ are single edges and length-$3$ paths; hence, $\maxden{G} \le \frac34$.
\end{itemize}
\end{claim}
\begin{proof}
We proceed by cases:
    \begin{itemize}
    \item Suppose $(\pi(x))_i = 1$. Then Alice creates the edges $\{(x,\vLeft,\vTop),(\pi(x),\vRight,\vTop)\}$ and $\{(x,\vLeft,\vBot),(\pi(x),\vRight,\vBot)\}$ while Bob creates the edges $\{(x,\vLeft,\vTop),(x,\vLeft,\vBot)\}$ and $\{(\pi(x),\vRight,\vTop),(\pi(x),\vRight,\vBot)\}$ (the latter using the assumption that $(\pi(x))_i = 1$. This shows the existence of a square. Bob only creates one vertical edge involving $\vLeft$ so therefore there is not more than one square. 
    \item For each $z \in \{0,1\}^n$, consider the set
    \[
    V^z \coloneqq \{(z,\vLeft,\vTop),(z,\vLeft,\vBot),(\pi(z),\vRight,\vTop),(\pi(z),\vRight,\vBot)\}.
    \]
    Note that every edge which is added by Alice or Bob to $G$ lies entirely within $V^z$ for some $z$.
Thus, the connected components of $G$ each lie entirely within $V^z$ for some $z$ and so it suffices to consider each $V^z$ individually.
Further, we never add edges connecting a $(\vLeft,\vTop)$ vertex and a $(\vRight,\vBot)$ vertex, nor connecting a $(\vLeft,\vBot)$ vertex and a $(\vRight,\vTop)$ vertex.
Thus, there are at most $4$ edges in the induced subgraph on $V^z$.
Now it is easy to check by the definition of Alice and Bob's edges that the only way all four edges are present is if $z=x$ and $(\pi(x))_i = 1$.
Otherwise, the induced subgraph on $V^z$ either consists of two parallel edges (if $z \ne x$ and $(\pi(z))_i = 0$) or is a path of length $3$ (if either $z=x$ or $(\pi(z))_i=1$). \qedhere
    \end{itemize}
\end{proof}

We now prove the same bound on outputting the solution. We again consider the same graph-theoretic interpretation of $\PermProb$, and Alice and Bob again create the same graph $G$ by running \Cref{alg:densub:sparse hard:alice-reduction,alg:densub:sparse hard:bob-reduction}. So, Alice inputs her edges into the stream, runs the algorithm that gives a $(1-\eps)$ approximate densest subgraph, and sends the algorithm's memory state to Bob. Then Bob inputs his edges, gets output $\calS$ (which is a set of vertices that induces a subgraph) and does the following: 
if $|\calS| \neq 4$, then he outputs $0$. If $\calS$ is of the form $\{(v_1,\vLeft,\vTop),(v_1,\vLeft,\vBot),(v_2,\vRight,\vTop),(v_2,\vRight,\vBot)\}$ where $v_1 \neq x$, he also outputs $0$. Else, he outputs $(v_2)_i$. 

 We now prove that Bob will be correct doing this protocol. We first consider the case where $(\pi(x))_i = 1$ and prove that Bob will output $1$. 
Recall \Cref{cl:denseSquare} which says that if $(\pi(x))_i = 1$, then $G$ contains exactly one square with the rest of the connected components being single edges and length-$3$ paths. 
\begin{claim} \label{cl:density of not square}
Consider a graph with a single $4$-cycle (i.e., a square) and all other connected components consisting of single edges and length-$3$ paths. Take the set of vertices in the $4$-cycle to be $\calC$. For any set $\calS \neq \calC$, $\den{G}{\calS} < 0.9$. 
\end{claim}
\begin{proof}
    We clearly have $\den{G}{\calC} = 1$. We now show that for any set $\calS \neq \calC$, $\den{G}{\calS} < 0.9$. We break into cases: 
    \begin{itemize}
        \item $\calS$ s.t. $|\calS| < |\calC|$: If we have $|\calS| = 1$, $\den{G}{\calS} = 0$. If we have $|\calS| = 2, \den{G}{\calS} \leq 1/2$ since the vertices can have at most one edge between them. If we have $|\calS| = 3, \den{G}{S} \leq 2/3$ since there exist no $3$-cycles in the graph. 
        \item $\calS$ s.t. $|\calS| = |\calC| = 4$: Since we have that $\calS \neq \calC$ and there is only one $4$-cycle in the graph, we have $\den{G}{\calS} \leq 3/4$. 
        \item $\calS$ s.t. $|\calS| > |\calC|$:  WLOG, we can consider $\calS$ to be $\calC$ with additional vertices added.
        \begin{itemize}
            \item Adding a single vertex adds no edges. Therefore, adding any number of single vertices gives us $\den{G}{\calS} \leq 4/5$. 
            \item Adding an edge adds two vertices and one edge to the subgraph. Therefore, adding any number of edges gives us $\den{G}{\calS} \leq 5/6$. 
            \item Adding a length-$3$ path adds four vertices and three edges. Therefore, adding any number of length-$3$ paths gives us $\den{G}{\calS} \leq 7/8$. 
        \end{itemize}
        Adding some combination of single vertices, edges, and length-$3$ paths to $\calC$ to form $\calS$ gives us $\den{G}{\calS} \leq 7/8$. 
    \end{itemize}
\end{proof}
Therefore, combining \Cref{cl:denseSquare} and \Cref{cl:density of not square}, we have that when $(\pi(x))_i = 1$, there will be a single $4$-cycle in the graph and the set of participating vertices $\calS$ will be outputted as the approximately optimal densest subgraph for $\eps \leq 0.1$. So, we will have $|\calS| = 4$, by construction $\calS$ will be of the form $\{(x,\vLeft,\vTop),(x,\vLeft,\vBot),(v_2,\vRight,\vTop),(v_2,\vRight,\vBot)\}$, and $(v_2)_i = 1$ which is what Bob will output. 

We now consider the case that $(\pi(x))_i = 0$ and prove that Bob will output $0$. By \Cref{cl:denseSquare}, the only connected components in the graph will be edges and length-$3$ paths. Since Bob outputs $0$ unless $|\calS| = 4$, the only two cases we have to consider are the following: 
\begin{itemize}
    \item $|\calS|= 4$ and $\calS$ induces a subgraph which includes two (disjoint) edges:  $\maxden{G} \geq 3/4$ since there exist many length-$3$ paths in the graph. One that must exist is the one which involves $(x, \vLeft, \vTop)$
 and $(x, \vLeft, \vBot)$. In this case, $\den{G}{S} = 1/2$. Therefore, for $\eps \leq 0.1$, $\calS$ cannot be outputted as an approximately optimal densest subgraph. 
 \item $|\calS| = 4$ and $\calS$ induces a subgraph which is a length-$3$ path: If $\calS$ is not of the form 
 \[\{(x,\vLeft,\vTop),(x,\vLeft,\vBot),(v_2,\vRight,\vTop),(v_2,\vRight,\vBot)\}\] then Bob outputs $0$. If $\calS$ is of that form, the three edges by construction must be \[\{(x,\vLeft,\vTop),(x,\vLeft,\vBot)\}, \{(x,\vLeft,\vTop), (v_2,\vRight,\vTop)\},\] and $\{ (x,\vLeft,\vBot), (v_2,\vRight,\vBot)\}$. However, then we have $(\pi(x))_i = (v_2)_i = 0$ which is what Bob outputs. 
 \end{itemize}

 \subsection{Upper bound for dense graphs}
 \begin{theorem}\label{thm:densub:dense alg}
There exists an insertion-only streaming algorithm which, given an undirected graph $G = (V, E)$ on $n$ vertices and $m \ge \alpha n^2$ edges and $\eps \in (0,1)$, outputs $S \subseteq [n]$ which has $\den{G}{S} \ge (1-\eps) \cdot \maxden{G}$ in $O(\frac{n}{\eps^2\alpha^4})$ bits of space with probability at least $2/3$. 
\end{theorem}
\begin{proof}
Similar to the algorithm for $\MaxCut$, we apply the $F_0$ estimation sketch of \cite{B2018optimal} to get the result. Again the universe is the set of all possible edges of the input graph $G$. 

During the stream, we maintain the exact edge count $m = |E|$ using an $O(\log n)$ bit counter. We also feed the stream of inserted edges into the $F_0$ sketch, instantiating it with $\eps' = \frac{\eps \alpha^2}{4}$ and failure probability $\delta = \frac{1}{2 \cdot 2^n}$. The total space used is therefore $O\left(\frac{n}{\eps^2\alpha^4}\right)$ bits.

For any subset $S \subseteq V$, let $\calD(x)$ denote the set of $\binom{|S|}{2}$ edges in the completed induced subgraph on $S$. Then, the number of edges in $G[S]$ is $|E \cap \calD(x)|$. By inclusion-exclusion, we have 
\[
|E \cap \calD(x)| = m + |\calD(x)| - |E \cup \calD(x)|. 
\]
After the stream, we look at all $2^n$ possible subsets $S$ of $V$. For each subset, we instantiate a copy of the $F_0$'s memory state and feed the offline edges of $\calD(x)$ into it. The sketch then outputs estimate $\hat{U}_S$. So, we estimate the number of induced edges as $\hat{E}_S = m + |\calD(S)| - \hat{U}_S$, and the density as $\hat{D}_S = \hat{E}_S / |S|$.

With probability at least $2/3$, for all $2^n$ subsets simultaneously, the additive error in estimating $|E \cap \calD(S)|$ is at most $\eps' \binom{n}{2} \le \frac{\eps\alpha^2 n^2}{8}$.
We restrict our search to subsets $S$ where $|S| \ge m/n \ge \alpha n$, as any optimal subgraph must have at least this many vertices (since $\maxden{G} \ge \den{G}{V} = m/n \ge \alpha n$). For these valid subsets, the additive error in the estimated density is at most $\frac{\eps\alpha^2 n^2}{8 |S|} \le \frac{\eps\alpha^2 n^2}{8 \alpha n} = \frac{\eps\alpha n}{8} \le \eps \maxden{G}$. Outputting the valid subset $S$ that maximizes $\hat{D}_S$ provides a $(1-\eps)$-approximation.
\end{proof}

We now provide our alternate upper bound. 
\begin{theorem}
There exists an insertion-only streaming algorithm which, given an undirected simple graph $G = (V, E)$ on $n$ vertices and $m \ge \alpha n^2$ edges and $\eps \in (0,1)$, outputs $S \subseteq [n]$ which has $\den{G}{S} \ge (1-\eps) \cdot \maxden{G}$ in $O(\frac{n}{\eps^2\alpha^4} \log \frac{1}{\eps \alpha^2})$ bits of space and $O(\log n)$ amortized update time with probability at least $2/3$. 
\end{theorem}
\begin{proof}
In this section we present a $O(n)$ space algorithm for $\DenSub$ in dense graphs using \Cref{thm:sampAlg} to prove \Cref{thm:densub:dense alg}. 
Let us call $\eps' = \eps \alpha / 16$.
We set $N = \binom{n}{2}$, $t = 4n/(\eps'^2 \alpha^2)$, and $\lambda = \eps' \alpha$.
So the space is $O(\log(1/(\eps\alpha^2)) \cdot n /(\eps^2 \alpha^4))$ and the amortized update time is $O(\log n)$.

For some subset $S \subseteq V$, we want to evaluate function $f_S : \binom{[n]}2 \to \{0,1\}$ where $f_S(\{u,v\}) = 1$ if $u \in S$ and $v\in S$ and $0$ otherwise.
The output of \Cref{thm:sampAlg} for function $f_S$ is $\hat{\bmu}(f_S)$, which is an estimate of
\[
\mu_{E}(f_S) \coloneqq \Exp_{\bw \sim \Unif{E}} f_S(\bw) = \frac{|S|}m \cdot \den{G}{S}.
\]
To get our final estimate for $\den{G}{S}$ we output $(m / |S|) \cdot \hat{\bmu}(f_S)$.
Note that we can easily get value $m$ by keeping an exact counter in the stream which takes at most $O(\log n)$ space.
To find the densest subset, we estimate $\den{G}{S}$ for each $S$ such that $|S| \ge \frac{m}{n}$ and then output the densest. 

Now we proceed with the analysis.
Note that $\mathcal{M}$ from \Cref{thm:sampAlg} is the edge set $E$.
Fix a subset $S \subseteq V$.
By the guarantee from \Cref{thm:sampAlg} we have that
\[
\Pr_{\bsigma} \bracks*{ \abs*{ \hat{\bmu}(f_S) - \mu_E(f_S) } \ge  2\eps' } \le \frac{1}{3} \cdot \exp(-\Omega(n)).
\]
Note that $\den{G}{S} = \sum_{w \in E} f_S(w) /|S|$.
So, for fixed $S$, we have:
\[
\Pr_\bsigma \bracks*{ \abs*{ \frac{m}{|S|} \cdot \hat{\bmu}(f_S) - \den{G}{S} } \le \frac{m}{|S|}  \cdot 2 \eps ' } \ge (2/3) \cdot \exp(-\Omega(n)) .
\]

Now we will show that $(m/|S|) \cdot 2 \eps' \le \eps \cdot \maxden{G}$. 

\begin{claim}\label{claim:densub:maxden-lb}
    Let $G=(V,E)$ be an undirected multigraph on $n$ vertices with $m$ edges.
Then \[
    \maxden{G} \ge \frac{m}n.
    \]
\end{claim}

\begin{proof}
    By definition, $\maxden{G} \ge \den{G}{V} = \frac{m}n$.
\end{proof}

\begin{claim}\label{claim:densub:S-m/n-lb}
    Let $G=(V,E)$ be an undirected multigraph on $n$ vertices with $m$ edges.
    Then \[
    \abs*{ \argmax_{\substack{S \subseteq V \\ S \ne \emptyset}} \den{G}{S} } \ge \frac{m}n.
    \]
\end{claim}

\begin{proof}
    Let $S$ be any set s.t. $\den{G}{S} = \maxden{G}$.
    Then \[ \den{G}{S} = \frac{|G[S]|}{|S|} \le \frac{|S|^2}{|S|} = |S|. \]
    Since $\maxden{G} \ge \frac{m}n$ by \Cref{claim:densub:maxden-lb} we deduce $|S| \ge \frac{m}n$ as desired.
\end{proof}

Recall that we consider only $S$ such that $|S| \ge \frac{m}{n}$.
By \Cref{claim:densub:S-m/n-lb} we will not miss the densest subgraph.
So we have that 
\[
\frac{m}{|S|} \cdot 2\eps' =  \frac{m}{|S|} \cdot \frac{\eps \alpha}{8} \le \frac{m}{|S|} \cdot \frac{\eps}{8} \cdot \frac{m}{\binom{n}{2}} \le n \cdot \frac{\eps}{8} \cdot \frac{m}{\binom{n}{2}} = \frac{\eps}{8} \cdot \frac{2m}{n-1} \le \frac{\eps}{2} \cdot \maxden{G}.
\]
Taking a union bound over all possible $S$ such that $|S| \ge \frac{m}{n}$, we have that
\begin{equation} \label{eq:presS}
\Pr_\bsigma \bracks*{ \forall S \subseteq V, |S| \ge \frac{m}{n}, \quad \abs*{ \frac{m}{|S|} \cdot \hat{\bmu}(f_S) - \den{G}{S} } \le \eps \cdot \maxden{G} } \ge \frac23.
\end{equation}

Now, condition on any fixed values of $\hat{\mu}(S)$ satisfying the conditions of \Cref{eq:presS} and define \[
S^* \coloneqq \argmax_{\substack{S \subseteq V \\ S \ne \emptyset}} \den{G}{S} \text{ and } S' \coloneqq \argmax_{\substack{S \subseteq V \\ |S| \ge \frac{m}{n}}} \hat{\mu}(f_{S'}) . \]
So we have that 
\begin{align*}
    \den{G}{S'} + \frac{\eps}{2} \cdot \maxden{G} \tag{by \cref{eq:presS}} &\ge \frac{m}{|S|}\cdot \hat{\mu}(f_{S'}) \\
    &\ge \frac{m}{|S|}\cdot \hat{\mu}(f_{S^*}) \tag{since $S'$ maximizes $\hat{\mu}(f_{S'})$}\\
    &\ge \den{G}{S^*} - \eps \cdot  \maxden{G}.  \tag{by \cref{eq:presS}}
\end{align*}
Since $\maxden{G} = \den{G}{S^*}$, we are done.
\end{proof}
 \subsection{Randomized lower bound for dense graphs}

 \begin{theorem}\label{thm:densub:dense hard}
For constant $\eps > 0$, any randomized streaming algorithm which outputs a set $S \subseteq V$ such that $\den{G}{S} \ge (1-\eps) \cdot \maxden{G}$ for input graph $G$ on $n$ vertices and $m = \Theta(n^2)$ edges with probability at least $2/3$ requires $\Omega(n)$ bits of working space. 
\end{theorem}

 Now, we prove \Cref{thm:densub:dense hard}.
Our proof follows the proof of \Cref{thm:maxcut:dense hard}; by the exact same argument as in that section, to prove \Cref{thm:densub:dense hard} it suffices to prove the following:

\begin{definition}\label{def:densub:dense hard:disjoint}
    Let $\gamma \in [0,1]$.
    Two graphs $G$ and $G'$ on $[n]$ are \emph{$\gamma$-disjoint} if for every $T \subseteq [n]$,
    it is not simultaneously the case that $\den{G}{T} \ge \gamma \maxden{G}$ and $\den{G'}{T} \ge \gamma \maxden{G'}$.
\end{definition}

Again, we have a theorem asserting the existence of a large family of ``disjoint'' graphs.
\begin{theorem}\label{thm:densub:dense hard:instances}
    There exists $\alpha, C > 0$ such that the following holds.
    For $\epsilon_0 \coloneqq \tfrac14$ and every $0 < \epsilon < \epsilon_0$, let $n \in \N$ be sufficiently large.
    Then there exists a set $\calG$ such that:
    \begin{enumerate}
        \item Each $G \in \calG$ is an undirected simple graph on $[n]$ with $\ge \alpha \binom{n}2$ edges.
        \item $|\calG| \ge 2^{Cn/\eps}$.
        \item For every $G \ne G' \in \calG$, $G$ and $G'$ are $(1-\eps)$-disjoint (in the sense of \Cref{def:densub:dense hard:disjoint}).
    \end{enumerate}
\end{theorem}

Indeed, for $S \subseteq [n]$ let $\kom{S}$ denote the graph on $[n]$ which is a clique on $S$, so $\kom{S}$ has $\binom{|S|}2$ edges and \[ \maxden{\kom{S}} = \den{\kom{S}}{S} = \frac{\binom{|S|}2}{|S|} = \frac{|S|-1}2. \]

We claim the following about densities of subgraphs of $\kom{S}$:

\begin{claim}
    Let $0 < \eps < \frac12$. Suppose $T \subseteq [n]$ satisfies $\den{\kom{S}}T \ge (1-\eps) \maxden{\kom{S}}$. Then $d(1_S,1_T) \le 4\eps$, where $d(\cdot,\cdot)$ denotes the normalized Hamming distance.
\end{claim}

\begin{proof}
    Observe that the induced subgraph of $\kom{S}$ on $T$ is simply $\kom{S \cap T}$. Hence
    \[
    \den{\kom{S}}{T} = \frac{\binom{|S \cap T|}2}{|T|} = \frac{(|S \cap T|(|S \cap T|-1)}{2 |T|},
    \]
    hence
    \begin{equation}\label{eq:densub:dense hard}
    (1-\eps) (|S|-1) \le \frac{|S \cap T|(|S \cap T|-1)}{|T|}.
    \end{equation}
    First, we deduce
    \[
    (1-\eps) (|S|-1) \le \frac{|T|(|S \cap T|-1)}{|T|} = |S \cap T|-1,
    \]
    and therefore $(1-\eps)|S| \le |S \cap T|$. We similarly calculate \[
    (1-\eps) (|S|-1) \le \frac{|S|(|S|-1)}{|T|}
    \]
    and therefore $|T| \le \frac1{1-\eps}|S| \le (1+2\eps)|S|$ assuming WLOG that $\eps < \frac12$.

    Now we deduce \[ (1-\eps)|S| \le |S \cap T| = |S| - |S \cap \compT|, \] so $|S \cap \compT| \le \eps |S|$. Similarly, \[ |\compS \cap T| + (1-\eps)|S| \le |\compS \cap T| + |S \cap T| = |T| \le (1+2\eps) |S|. \] Hence $|\compS \cap T| \le 3\eps |S|$.

    Finally, we conclude
    \[ d(1_S,1_T) = \frac1n(|S \cap \compT| + |\compS \cap T|) \le \frac{4\eps |S|}{n} \le 4\eps,
    \]
    as desired.
\end{proof}

Hence \Cref{thm:densub:dense hard:instances} follows from the existence of an exponentially large set of $n$-bit strings the radius-$4\eps$ balls around which are pairwise disjoint; this follows from a standard packing argument or from the probabilistic method and is also implied by the stronger \Cref{lemma:maxcut:dense hard:counting}.

\section{Extensions} \label{sec:extensions}
\subsection{All CSPs}
We first extend the upper bound for $\MaxCut$ to all other $\textsc{CSP}$s with constant alphabet size and arity. 

$\MaxCSP$ is an abstraction of the widely-studied \emph{constraint satisfaction problem}.
Let $k,q \in \N$ and let $V$ be a finite set of \emph{variables}.
A \emph{constraint} is a pair $(v,\Pi)$, where $v : [k] \to V$ is an injective function (i.e., $v(i) \ne v(j)$ for $i \ne j \in [k]$) and $\Pi : [q]^k \to \{0,1\}$ is a function called the \emph{predicate}.
An \emph{assignment} is a function $x : V \to [q]$, and $x$ \emph{satisfies} the constraint $(v,\Pi)$ if $\Pi(x(v(1)),\ldots,x(v(k)))=1$.
An \emph{instance} $\Phi$ of $\MaxCSP$ is a set of constraints, and the \emph{value} of an assignment on an instance is the probability a random constraint is satisfied:
\[
\val{\Phi}{x} \coloneqq \Exp_{(\bv,\bPi) \sim \Unif\Phi}[\bPi(x(\bv(1)),\ldots,x(\bv(k)))]
\]
and the value of an instance is the maximum value of any assignment:
\[
\maxval{\Phi} \coloneqq \max_{x : V \to [q]} \val{\Phi}{x}.
\]

$\MaxCut$ is a special case of $\MaxCSP$ where $k=q=2$ and every predicate is $\Pi : \{0,1\}^2 \to \{0,1\}$ with $\Pi(y_1,y_2) = 1$ iff $y_1 \ne y_2$.
Other problems captured by $\MaxCSP$ include maximum directed cut, maximum unique games, and many others. 

\begin{theorem}\label{thm:maxcsp:dense alg}
Let $k,q \in \N$.
There exists an insertion-only streaming algorithm which, given an instance $\Phi$ on $n$ vertices with $m \ge \alpha n^k$ constraints and $\eps \in (0,1)$, outputs $x : V \to [q]$ which has $\val{\Phi}{x} \ge (1-\eps) \cdot \maxval{\Phi}$ in $O(\frac{n}{\eps^2\alpha^2})$ bits of space with probability at least $2/3$.
\end{theorem}
\begin{proof}
Recall that $k$ and $q$ are both constants. 
Our universe $\calU$ is all the possible constraints. So, the universe size is at most $n^k \cdot 2^{q^k}$. Similarly to for $\MaxCut$, we maintain the exact constraint count $m$ with an $O(\log n)$ bit counter. Then we feed the stream of constraints into the $F_0$ sketch which is instantiated with $\eps' = \frac{\eps \alpha}{8 \cdot 2^{q^k}\cdot q^k}$ and $\delta = \frac{1}{3 \cdot q^n}$. 
The space used is therefore $O\!\left(\frac{n}{\eps^2\alpha^2}\right)$ bits. 

For any assignment $x:V\to[q]$, let $T(x)\subseteq \mathcal{U}$ denote the set of
constraints that are satisfied by $x$, i.e.,
\[
T(x)\ :=\ \{(v,\Pi)\in\mathcal{U}:\ \Pi(x(v(1)),\ldots,x(v(k)))=1\}.
\]
The number of constraints of $\Phi$ satisfied by $x$ is exactly $|\Phi\cap T(x)|$.
By inclusion--exclusion,
\[
|\Phi\cap T(x)| \;=\; m + |T(x)| - |\Phi\cup T(x)|.
\]

After the stream, we iterate over all $q^n$ assignments $x$. For each $x$, we
instantiate a copy of the $F_0$ sketch's memory state and feed the offline
constraints of $T(x)$ into it. The sketch then outputs an estimate $\hat{\bU}_x$
for $|\Phi\cup T(x)|$. We estimate the number of satisfied constraints by
\[
\hat{\bV}_x \;:=\; m + |T(x)| - \hat{\bU}_x,
\qquad
\widehat{\val{\Phi}{x}} \;:=\; \frac{\hat{\bV}_x}{m}.
\]
By a union bound, with probability at least $2/3$, for all $q^n$ assignments $x$ simultaneously, we have 
\[
\abs{\hat \bV_x - |\Phi\cap T(x)|}
=
\abs{\hat \bU_x - |\Phi\cup T(x)|}
\le
\eps'\,|\Phi\cup T(x)|
\le
\eps' |\calU|.
\]
Therefore the normalized values satisfy
\[
\abs{\widehat{\val{\Phi}{x}} - \val{\Phi}{x}}
\;\le\;
\frac{\eps' |\calU|}{m}
\;\le\;
\eps'\cdot \frac{n^k 2^{q^k}}{\alpha n^k}=
\frac{\eps}{8 q^k}.
\]

Let $x^\star$ be an optimal assignment such that $\val{\Phi}{x^\star}=\maxval{\Phi}$ and let
$\hat x\in\arg\max_x \widehat{\val{\Phi}{x}}$. Then
\[
\val{\Phi}{\hat x}
\ \ge\
\widehat{\val{\Phi}{\hat x}}-\frac{\eps}{8q^k}
\ \ge\
\widehat{\val{\Phi}{x^\star}}-\frac{\eps}{8q^k}
\ \ge\
\maxval{\Phi}-\frac{\eps}{4q^k}.
\]
A uniformly random assignment satisfies each
constraint with probability at least $1/q^k$, so we have $\maxval{\Phi} \ge 1/q^k$.
Therefore, we have $\eps / (4 q^k) \leq \eps \cdot \maxval{\Phi}$, and we have the result.  
\end{proof}

We again give our alternate algorithm with faster update time. 

\begin{theorem} \label{thm:altForCSP}
Let $k,q \in \N$.
There exists an insertion-only streaming algorithm which, given an instance $\Phi$ on $n$ vertices with $m \ge \alpha n^k$ constraints and $\eps \in (0,1)$, outputs $x : V \to [q]$ which has $\val{\Phi}{x} \ge (1-\eps) \cdot \maxval{\Phi}$ in $O(\frac{n}{\eps^2\alpha^2} \log \frac{1}{\eps \alpha})$ bits of space and $O(\log n)$ amortized update time with probability at least $2/3$.
\end{theorem}
\begin{proof}
Our algorithm builds on \Cref{thm:sampAlg}.
Fix $k$ and $q$.
Set $\eps' \coloneqq \frac{\eps}{4q^k}$ and $\lambda = \eps' \alpha$, so that $\eps' + \lambda \alpha = \frac{\eps}{2q^k}$.
We set $N = n^{\underline k} 2^{q^k}$ and identify $[N]$ with the space $\calC = \{(v,\Pi)\}$ of all possible constraints.\footnote{
    Here $n^{\underline k}$ denotes the \emph{falling factorial}, i.e., $n(n-1)\cdots(n-k+1)$.
    This counts the number of injections from $[k]$ to $[n]$.}
We set $t = 4n/(\eps'^2 \alpha^2)$, and $\lambda = \eps' \alpha$. So we get the desired space and amortized update time.  

Our algorithm is as follows: We view the stream of constraints $\Phi$ as a multiset $\calM$ which is a subset of the (set) $\calC$ in the sense of \Cref{thm:sampAlg}.
For each assignment $x : V \to [q]$, we have a function $f_x : \calC \to \{0,1\}$ which outputs $1$ on a constraint $(v,j)$ iff $x$ satisfies $(v,\Pi)$, i.e., if $\Pi(x(v(1)),\ldots,x(v(k)))=1$.
The algorithm in \Cref{thm:sampAlg} gives an estimate $\hat{\mu}(f_x)$ of 
\[
\mu_\calM(f_x) := \Exp_{(\bv,\bPi)\sim\Unif{\Phi}} [f_x(\bv,\bPi)] = \val{\Phi}{x}
\]
for every $x$.
Finally, we output the assignment $x' \coloneqq \argmax_{x : V \to [q]} \hat{\mu}(f_x)$.

Now we proceed with the analysis.

\begin{proof}[Proof of \Cref{thm:altForCSP}]
By the guarantee in \Cref{thm:sampAlg}, we have that for every fixed $x : V \to [q]$,
\[
\Pr_\bsigma\left[ \abs*{ \hat{\bmu}(f_x) - \val{\Phi}{x} } \ge \frac{\eps}{2q^k} \right] \le \frac{2}{3} \cdot \exp(-\Omega(n)).
\]

There are $q^n$ possible assignments $x : V \to [q]$, each with a corresponding function $f_x : \calC \to \{0,1\}$.
So we have that
\[
\Pr_{\bsigma} \bracks*{ \forall x : V \to [q], \quad \abs* {\hat{\bmu}(f_x) - \val{\Phi}{x} } \le \frac{\eps}{2q^k} } \ge \frac23 \tag{$\star$}
\]
holds.
Condition on fixed values of $\hat{\mu}(f_x)$ satisfying $(\star)$ for the remainder of the proof.

Now finally let $x^* \coloneqq \argmax_{x : V \to [q]} \mu_\calM(f_x)$ (so that $\val{\Phi}{x^*} = \maxval{\Phi}$).
So we have 
\begin{align*}
\val{\Phi}{x'} + \frac{\eps}{2q^k} & \ge \hat{\mu}(f_{x'}) \tag{$\star$} \\ 
&\ge \hat{\mu}(f_{x^*}) \tag {since $x'$ maximizes $\hat{\mu}(f_{x'})$} \\ 
&\ge \val{\Phi}{x^*} - \eps' \tag{$\star$} \\
&= \maxval{\Phi} - \frac{\eps}{2q^k}.
\end{align*}

So we have that 
\[
\val{\Phi}{x'} \ge \val{\Phi}{x^*} -  \frac{\eps}{q^k}.
\]
Recall that $\val{\Phi}{x^*} \ge \frac{\eps}{q^k}$.
So, we are done.
\end{proof}
\end{proof}

\subsection{Similarity and rarity}
We give our algorithms for $\Sim$ and $\Rare{k}$ here. 
Given two datasets made up of elements from $[N]$, similarity is a measure of how similar two datasets are and is useful for estimating transitive closures \cite{COHEN1997441}, web page duplication detection \cite{broder00identifying}, and data mining \cite{cohen01finding}.
\textcite{datar02estimating} give an insertion-only algorithm using $O(1/\eps^2 (\log N + \log m))$ bits of space for an additive $\pm \eps$ approximation, where $m$ is the total length of the stream.
\cite{FEIGENBLAT2017171} improve this with an insertion-only streaming algorithm using $O(\log N/\eps^2)$ bits of space.
Note that in their paper they specify their space in words.
We further improve this with an algorithm that only uses $O(1/\eps^2 + \log N)$ \emph{bits} of space. Interestingly, we do not require the stream to be dense. 

Given a stream of elements from $[N]$, $\Rare{k}$ is the problem of estimating the number of elements (out of $Q$ total distinct elements) that occur $k$ times in a dataset.
This is a useful quantity which can be used to compute the value of any symmetric function on the frequency of stream elements.
In particular, it can be used to estimate the number of distinct elements \cite{FLAJOLET1985182}, frequency moments \cite{ALON1999137}, capped statistics of a stream \cite{cohen15stream}, the objective function of $M$-estimators \cite{jayaram2021trulyperfectsamplersdata}, and for applications including computing degree distributions in large graphs \cite{buriol05using}, detecting malicious IP traffic in a network \cite{karamchetti05detecting}, and other various problems in databases \cite{cormode05summarizing}.
(See \cite{chen2023spaceoptimalprofileestimationdata} for more.)
The first work to compute rarity in a stream was by \cite{datar02estimating} which gave a insertion-only algorithm using $O(1/\eps^2 (\log N + \log m))$ bits of space for an additive $\pm \eps$ approximation where $m$ is the total length of the stream.
\cite{chen2023spaceoptimalprofileestimationdata} improved this to $O(1/\eps^2 + \log N)$ bits of space.
They also give a $O(1/\eps^2 \log(1/\eps) + \log N+ \log \log m)$-space algorithm for a $\pm \eps \frac{m}Q$ approximation.
In the dense case, we match their first algorithm's space of $O(1/\eps^2 + \log N)$.
However, we are able to present a substantially simpler analysis.

\begin{theorem}\label{thm:sim-alg}
For every $\eps> 0$ and $N \in \N$, there exists an insertion-only streaming algorithm which, given as input a sequence of insertions to two sets $\calA, \calB \subseteq [N]$, outputs \[ v \in \frac{|\calA \cap \calB|}{|\calA \cup \calB|} \pm \eps \] with probability at least $2/3$ in $O(1/\eps^2 + \log N)$ bits of space.
\end{theorem}
\begin{proof}
By inclusion-exclusion, $|\calA \cap \calB| = |\calA| + |\calB| - |\calA \cup \calB|$.
We instantiate three independent copies of the $F_0$ algorithm of \cite{B2018optimal} with $\eps' = \frac{\eps}{4}$ and failure probability $\delta = \frac{1}{9}$. The first sketch processes the insertions to $\calA$, the second processes $\calB$, and the third processes both (representing $\calA \cup \calB$). The total space required is $O\left(\frac{1}{\eps^2} + \log N\right)$.

With probability at least $2/3$, all three sketches succeed, providing estimates $\hat{F}_\calA$, $\hat{F}_\calB$, and $\hat{F}_{\calA \cup \calB}$ with relative error at most $\eps'$. We estimate the similarity as 
\[\frac{\hat{F}_\calA + \hat{F}_\calB - \hat{F}_{\calA \cup \calB}}{\hat{F}_{\calA \cup \calB}}.\] 
The absolute additive error in the numerator is bounded by $\eps'(|\calA| + |\calB| + |\calA \cup \calB|) \le 3\eps' |\calA \cup \calB|$. Since the denominator satisfies $\hat{F}_{\calA \cup \calB} \ge (1-\eps')|\calA \cup \calB|$, the additive error of the ratio is at most $\frac{3\eps'}{1-\eps'} \le 4\eps' = \eps$.
\end{proof}

We again give our alternate algorithm for $\Sim$ with faster update time. We first require some set-up and definitions.

A family of functions $\calH \subseteq U^{[S]}$ is called \emph{pairwise independent} if for all $i \neq j \in U$ and $a,b \in [S]$, we have 
\[
\Pr_{\bh \sim \Unif{\calH}}[\bh(i) = a \text{ and } \bh(j) = b] = \frac{1}{S^2}. 
\]
Storing $h$ from $\calH$ requires $O(\log |U|)$ bits. We require the following variant family of hash functions: 
\begin{lemma}[$2$-wise independent hash permutation family]\label{prelim:twowise}
    Let $N$ be a prime.
    Let $x_1\ne x_2 \in [N]$ and $y_1\ne y_2 \in [N]$.
    For $\bc \sim \Unif{\{1,\ldots,N-1\}}$ and $\bd \sim \Unif{\{0,\ldots,N-1\}}$, letting $h_{\bc,\bd} : [N] \to [N] : x \mapsto \bc x + \bd$, 
    \[
    \Pr_{\bc,\bd}[h_{\bc,\bd}(x_1)=y_1 \wedge h_{\bc,\bd}(x_2)=y_2] = \frac{1}{N(N-1)}.
    \]
\end{lemma}
In other words, for fixed $x_1 \ne x_2 \in [N]$, the marginal distribution $(h_{\bc,\bd}(x_1),h_{\bc,\bd}(x_2)) \sim \Unif{\{(y_1,y_2) : y_1 \ne y_2 \in [N]\}}$.
In particular, $h_{c,d}$ is always a bijection.

Now, we prove the following useful lemma. 
\begin{lemma}\label{lem:calc}
    Let $N$ be prime and let $h_{c,d} : [N] \to [N]$ be the random hash function in \Cref{prelim:twowise}.
For every set $\calS \subseteq [N]$ and $t \in \N$, for the random variable $X_\calS \coloneqq \abs*{ \braces*{ w \in \calS : h_{c,d}(w) \le t } }$, we have:
    \[
    \Exp X_\calS = \frac{|X| t}{N} \text{ and } \Var X_\calS \le \Exp X_\calS.
    \]
\end{lemma}

\begin{proof}
We have:
\[
X_\calS = \abs*{ \braces*{ w \in \calS : h_{c,d}(w) \le t } } = \sum_{w \in \calS} 1[h_{c,d}(w) \le t].
\]
Therefore by linearity of expectation and one-wise independence,
\[
\Exp X_\calS = \sum_{w \in \calS} \Pr[h_{c,d}(w) \le t] = \frac{|\calS| \cdot t}{N}.
\]
Further,
\begin{align*}
\Exp X_\calS^2 &= \sum_{w_1,w_2 \in \calS} \Pr[h_{c,d}(w_1) \le t \wedge h_{c,d}(w_2) \le t] \\
&= \sum_{w_1 \ne w_2 \in \calS} \Pr[h_{c,d}(w_1) \le t \wedge h_{c,d}(w_2) \le t] + \sum_{w \in \calS} \Pr[h_{c,d}(w) \le t] \\
&= \binom{|\calS|}2 \cdot \frac{t(t-1)}{N(N-1)} + \frac{|\calS| \cdot t}{N}
\end{align*}
by two-wise independence, and so
\begin{align*}
\Var X_\calS &= \Exp X_\calS^2 - \parens*{ \Exp X_\calS }^2 \\
&= \frac{|\calS| (|\calS|-1) \cdot t(t-1)}{N(N-1)} + \frac{|\calS| \cdot t}{N} - \parens*{ \frac{|\calS| \cdot t}{N} }^2 \\
&\le \frac{|\calS| \cdot t}{N} \\
&= \Exp X_\calS.
\end{align*}
(Note in particular that $\frac{|\calS|-1}N \le \frac{|\calS|}{N}$ since $|\calS| \le N$.)
\end{proof}
\begin{theorem}
For every $\eps, \alpha > 0$ and $N \in \N$, there exists an insertion-only streaming algorithm which, given as input a sequence of insertions to two sets $\calA, \calB \subseteq [N]$, assuming $|\calA \cup \calB| \ge \alpha N$, outputs \[ v \in \frac{|\calA \cap \calB|}{|\calA \cup \calB|} \pm \eps \] in $O(1/(\alpha\eps^2) + \log N)$ bits of space and $O(1)$ time (in arithmetic operations) per update with probability at least $2/3$.
\end{theorem}
\begin{proof}
 Now we present  \Cref{alg:sim}. 

\begin{algorithm}
    \caption{Algorithm for $\Sim$}\label{alg:sim}
    \begin{algorithmic}[1]
        \State \textbf{Before the stream:}
        \State Set $\delta \coloneqq \eps/3$ and $t \coloneqq \frac1{10 \delta^2 \alpha}$.
        \State Sample $c \sim \Unif{\{1,\ldots,N-1\}}$ and $d \sim \Unif{\{0,\ldots,N-1\}}$.
        \State Initialize two bit-arrays $a_1,\ldots,a_t$ and $b_1,\ldots,b_t$ to all zeroes.
        \State \textbf{During the stream:}
        \For{each update $u$}
        \If{$u$ inserts $w$ into $\calA$ and $h_{c,d}(w) \le t$}
        \State Set $a_{h_{c,d}(w)} \gets 1$.
        \ElsIf{$u$ inserts $w$ into $\calB$ and $h_{c,d}(w) \le t$}
        \State Set $b_{h_{c,d}(w)} \gets 1$.
        \EndIf
        \EndFor
        \State \textbf{After the stream:}
        \State Let $X_\cap \coloneqq \abs*{ \braces*{ 1 \le i \le [t] : a_i = 1 \wedge b_i = 1 } }$.
        \State Let $X_\cup \coloneqq \abs*{ \braces*{ 1 \le i \le [t] : a_i = 1 \vee b_i = 1 } }$.
        \Return $\frac{X_\cap}{X_\cup}$.
    \end{algorithmic}
\end{algorithm}

We first observe that \Cref{alg:sim} uses only $\frac1{\eps^2 \alpha} + \log N$ space (the first term for storing the bit arrays, the second for storing $c$ and $d$).

Note that $h_{c,d} : [N] \to [N]$ is always a permutation (for every pair $(c,d)$).
    Thus, for $i \le t$, $a_i = 1$ iff $h_{c,d}^{-1}(i) \in \calA$ and $b_i = 1$ iff $h_{c,d}^{-1}(i) \in \calB$.
Hence
\[
X_\cap = \abs*{ \braces*{ w \in \calA \cap \calB : h_{c,d}(w) \le t } }.
\]
Applying \Cref{lem:calc} gives $\Exp X_\cap = \frac{|\calA \cap \calB| \cdot t}{N}$ and $\Var X_\cap \le \Exp X_\cap$.
The same calculation gives $\Exp X_\cup = \frac{|\calA \cup \calB| \cdot t}{N}$ and $\Var X_\cup \le \Exp X_\cup$.

Now note that $\delta \Exp X_\cup \ge \delta \sqrt{\Exp X_\cup} \cdot \sqrt{\Exp X_\cap} \ge \delta \sqrt{X_\cup} \sqrt{\Var X_\cap}$.
Thus, by Chebyshev's inequality, we have:
\[
\Pr \bracks*{ \abs*{ X_\cap - \Exp X_\cap } \ge \delta \Exp X_\cup } \le \frac1{\delta^2 \Exp X_\cup} = \frac{N}{\delta^2 |\calA \cup \calB| t} \le \frac1{10},
\]
and similarly
\[
\Pr \bracks*{ \abs*{ X_\cup - \Exp X_\cup } \ge \delta \Exp X_\cup } \le \frac1{10}.
\]
Now condition on the event that 
\[
\abs*{ \frac{N}t X_\cap - |\calA \cap \calB|} \le \delta |\calA \cup \calB| \text{ and } \abs*{ \frac{N}tX_\cup - |\calA \cup \calB| } \ge \delta |\calA \cup \calB|; \tag{$\star$}
\]

We therefore have the upper bound
\[
\frac{X_\cap}{X_\cup} \le \frac{|\calA \cap \calB| + \delta |\calA \cup \calB|}{(1-\delta) |\calA \cup \calB|} = \frac{|\calA \cap \calB|}{|\calA \cup \calB|} + \parens* {\frac{1}{1-\delta} - 1} \frac{|\calA \cap \calB|}{|\calA \cup \calB|} + \frac{\delta}{1-\delta},
\]
and we can bound the error terms on the right-hand side by $\frac{1}{1-\delta} - 1 + \frac{\delta}{1-\delta} \le 3\delta$ for $\delta \le \frac13$.
Similarly,
\[
\frac{X_\cap}{X_\cup} \ge \frac{|\calA \cap \calB| - \delta |\calA \cup \calB|}{(1+\delta) |\calA \cup \calB|} = \frac{|\calA \cap \calB|}{|\calA \cup \calB|} + \parens* {\frac{1}{1+\delta} - 1} \frac{|\calA \cap \calB|}{|\calA \cup \calB|} - \frac{\delta}{1-\delta},
\]
and the right-hand side error is again at most $3\delta = \eps$.
\end{proof}   

Now we present our result for $\Rare{k}$. Unlike the previous problems, estimating rarity requires tracking exact element frequencies, which an $F_0$ sketch structurally discards. Thus we only present one algorithm based on our hashing method. 

\begin{theorem}\label{thm:rare-alg}
For every $\eps, \alpha > 0$ and $N,k \in \N$, there exists an insertion-only streaming algorithm which, given as input a sequence of insertions to a multiset $\calS \subseteq [N]$, letting $S_k$ denote the number of distinct elements in $\calS$ of multiplicity $k$ and $D$ the total number of distinct elements in $\calS$, assuming $D \ge \alpha N$, outputs \[ v \in \frac{S_k}{D} \pm \eps \] with probability at least $2/3$ in $O(\log k/(\alpha\eps^2) + \log N)$ bits of space and $O(1)$ time (in arithmetic operations) per update.
The multiplicity of an element can be unbounded. 
\end{theorem}
\begin{proof}
We first present \Cref{alg:rare}. 
\begin{algorithm}
    \caption{Algorithm for $\Rare{k}$}\label{alg:rare}
    \begin{algorithmic}[1]
        \State \textbf{Before the stream:}
        \State Set $\delta \coloneqq \eps/3$ and $t \coloneqq \frac1{10 \delta^2 \alpha}$.
        \State Sample $c \sim \Unif{\{1,\ldots,N-1\}}$ and $d \sim \Unif{\{0,\ldots,N-1\}}$.
        \State Initialize counters $a_1,\ldots,a_t$, each holding a value between $0$ and $k+1$.
        \State Initialize a bit-array $b_1,\ldots,b_t$.
        \State \textbf{During the stream:}
        \For{each update $u$ inserting $w$}
        \If{$h_{c,d}(w) \le t$}
        \If{$a_{h_{c,d}(w)} \le k$}
        \State Increment $a_{h_{c,d}(w)}$.
        \EndIf
        \State Set $b_{h_{c,d}(w)} \gets 1$.
        \EndIf
        \EndFor
        \State \textbf{After the stream:}
        \State Let $X_k \coloneqq \abs*{ \braces*{ 1 \le i \le [t] : a_i = k } }$.
        \State Let $X \coloneqq \abs*{ \braces*{ 1 \le i \le [t] : b_i = 1 } }$.
        \Return $\frac{X_k}{X}$.
    \end{algorithmic}
\end{algorithm}

The analysis of \Cref{alg:rare} is the same as the analysis of \Cref{alg:sim}, except for an extra $\log k$ factor needed to maintain the counters between $0$ and $k+1$.
(In particular, observe that $a_i = k$ at the end of the stream iff there are exactly $k$ copies of $h_{c,d}^{-1}(i)$ in $\calS$, and again $b_i = 1$ iff $h_{c,d}^{-1}(i)$ is in $\calS$.)

\end{proof}

\end{document}